\documentclass[11pt]{article}
\usepackage{amsmath,amssymb,amsfonts,amsthm,epsfig}
\usepackage[usenames,dvipsnames]{xcolor}
\usepackage{bm,xspace}
\usepackage{tcolorbox}
\usepackage{cancel}
\usepackage{fullpage}
\usepackage{liyang}
\usepackage{framed}
\usepackage{verbatim}
\usepackage{enumitem}
\usepackage{array}
\usepackage{multirow}
\usepackage{afterpage}
\usepackage{mathrsfs}
\usepackage{pifont} 
\usepackage{chngpage}
\usepackage[normalem]{ulem}
\usepackage{boxedminipage}
\usepackage{caption}
\usepackage{subcaption}
\usepackage{bold-extra}
\usepackage{forest}
\usepackage{etoolbox}
\usepackage{thm-restate}
\usepackage{microtype}
\usepackage{graphicx}

\usepackage{algorithm}
\usepackage{algorithmicx}
\usepackage{algpseudocode}

\usepackage{tikz}
\usetikzlibrary{calc,through,backgrounds,decorations.pathreplacing, calligraphy,arrows.meta}
\usetikzlibrary{positioning,chains,fit,shapes}
\usetikzlibrary{patterns}

\usepackage{tikz-cd}

\usepackage{pgfplots}
\pgfplotsset{width=8cm,compat=newest}

\def\colorful{0}

\ifnum\colorful=1
\newcommand{\violet}[1]{{\color{violet}{#1}}}

\newcommand{\red}[1]{{\color{red} {#1}}}

\fi
\ifnum\colorful=0
\newcommand{\violet}[1]{{{#1}}}

\newcommand{\red}[1]{{{#1}}}

\fi

\newcommand{\D}{\mathcal{D}}

\newcommand{\error}{\mathrm{error}}

\newcommand{\Adv}{\mathrm{Adv}}
\newcommand{\mistakes}{\mathrm{mistakes}}
\newcommand{\mrate}{\mathrm{mrate}}

\newcommand{\NP}{\mathsf{NP}}

\renewcommand{\P}{\mathsf{P}}

\newcommand{\bmcO}{\boldsymbol{\mcO}}

\newcommand{\VCdim}{\mathrm{VCdim}}

\newcommand{\Unif}{\mathrm{Unif}}

\newcommand{\pparagraph}[1]{\bigskip \noindent {\bf {#1}}}

\Crefname{theorem}{Theorem}{Theorems}
\Crefname{thmenumi}{Theorem}{Theorems}
\AtBeginEnvironment{theorem}{%
    \crefalias{enumi}{thmenumi}%
    \setlist[enumerate,1]{
        label={\textit{(\roman*)}},
        ref={\thetheorem(\roman*)}
    }%
}

\begin{document}
\title{
Samplability makes learning easier\\ \vspace{15pt}
}

\newcommand{\authcell}[2]{%
  \begin{tabular}[t]{@{}c@{}}#1\\[6pt]{\slshape #2}\end{tabular}%
}
 \author{%
 \hspace{2pt}
 \begin{tabular*}{\textwidth}{@{\extracolsep{\fill}}ccccc}
   \authcell{Guy Blanc}{Stanford} &
   \authcell{Caleb Koch}{Stanford} &
   \authcell{Jane Lange}{MIT} &
  \hspace{-10pt} \authcell{Carmen Strassle}{Stanford} &
 \hspace{-15pt}
   \authcell{Li-Yang Tan}{Stanford}
 \end{tabular*}
 \vspace{15pt}
 }

\date{\vspace{15pt}\small{\today}}

\maketitle

\begin{abstract}
    The standard definition of PAC learning (Valiant 1984) requires learners to succeed under all distributions---even ones that are intractable to sample from. This stands in contrast to samplable PAC learning (Blum, Furst, Kearns, and Lipton 1993), where  learners only have to succeed under samplable distributions. We study this distinction and show that samplable PAC substantially expands the power of efficient learners.
    
We first construct a concept class that requires exponential sample complexity in standard PAC but is learnable with polynomial sample complexity in samplable PAC. We then lift this statistical separation to the computational setting and obtain a separation relative to a random oracle. Our proofs center around a new complexity primitive, {\sl explicit evasive sets}, that we introduce and study. These are sets for which membership is easy to determine but are extremely hard to sample from.      
    
    Our results extend to the online setting to similarly show how its landscape changes when the adversary is assumed to be efficient instead of computationally unbounded. 
\end{abstract}




 \thispagestyle{empty}
 \newpage

 \thispagestyle{empty}
 \setcounter{tocdepth}{2}
 \tableofcontents
 \thispagestyle{empty}
 \newpage

 \setcounter{page}{1}

\section{Introduction}

In the PAC model~\cite{Val84}, the learner is given labeled data $(\boldsymbol{x},f(\boldsymbol{x}))$ where $\boldsymbol{x}$ is drawn from a distribution $\mathcal{D}$. \violet{The functions $f$ is promised to belong to a known, often simple, concept class, but no assumptions are made about $\mcD$.} Notably, the learner is required to succeed even under distributions that are  intractable to sample from. These are distributions $\mathcal{D}$ for which any generator $\mathcal{G}$ such that $\mathcal{G}(\mathrm{Unif}) = \mathcal{D}$ must have superpolynomial circuit size. Consequently, making even a single draw $\boldsymbol{x} \sim \mathcal{D}$ takes superpolynomial time. 

This is arguably an overly stringent requirement. PAC learners are expected to be efficient and yet the distribution, which can be viewed as an adversary in this context, is allowed to be computationally unbounded. This stacks the odds in favor of the adversary and is  one reason why efficient PAC learning algorithms have been hard to come by even for simple concept classes. Furthermore, if $\mathcal{D}$ is intractable to sample from then it alone renders the entire learning process inefficient, regardless of the efficiency of the learner. If the PAC model is to capture efficient end-to-end learning, from data collection to hypothesis generation, we may as well consider only samplable distributions. Relatedly, if one believes in the strong Church-Turing thesis, then all distributions occurring in practice are samplable. Quoting~\cite{Imp95fiveworlds}, ``Presumably, real life is not so adversarial that it would solve intractable problems just to give us a hard time." 

\paragraph{Samplable PAC.} It is therefore natural to consider a variant of the PAC model, samplable PAC, where the distribution is assumed to be samplable but is otherwise still unknown and arbitrary. This strikes the balance of imposing just enough structure on the distribution to place the learner and adversary on equal footing, while still allowing for enough expressivity to capture the complexities of real-world learning.

 The samplable PAC model was first considered by Blum, Furst, Kearns, and Lipton~\cite{BFKL93}. \violet{The focus of their paper was not on the distinction between samplable and standard PAC.} Rather, they showed that efficient learning in samplable PAC is tightly connected to the existence of  fundamental cryptographic primitives. As their main result, they proved that if one-way functions do not exist, then every concept class is {\sl average-case} learnable in samplable PAC.\footnote{This is a further relaxation of samplable PAC where there is an additional distribution, this one over target functions, and the learner is only required to succeed with respect to a random target function drawn from this distribution. We do not consider this variant in our work.} 




\subsection{This work}

\violet{We study the distinction between samplable and standard PAC. 
We are interested in formalizing the extent to which the assumption of samplability---a seemingly mild and reasonable assumption---expands the power of efficient learners. In measuring efficiency, we focus on the two most basic resources in learning: samples and runtime.}


\paragraph{Statistical separation.}  
The sample complexity of learning in standard PAC is fairly well-understood, in large part due to an elegant characterization in terms of VC dimension~\cite{VC71,BEHW89}---a foundational result now called ``The Fundamental Theorem of PAC Learning"~\cite{SSBD14}.  \violet{A VC dimension lower bound is generally viewed as an information-theoretic no-go in terms of efficient learnability: If a learning task cannot be learned with a reasonable number of samples, that trivially implies that it cannot be learned in a reasonable amount of time either.} 

Our first result is as follows: 
\smallskip 

 \begin{tcolorbox}[colback = white,arc=1mm, boxrule=0.25mm]
\begin{theorem}[See~\Cref{thm:statistical formal} for the formal version]
    \label{thm:statistical intro}
    There is a concept class with exponential VC dimension---and hence requires exponential sample complexity in standard PAC---but is learnable with polynomial sample complexity, \violet{and in fact even in polynomial time}, in samplable PAC. 
\end{theorem}
\end{tcolorbox}
\medskip

This shows that VC dimension lower bounds can be overly pessimistic. A learning task with large VC dimension may nevertheless be efficiently learnable if one assumes that these samples are generated according to a reasonable distribution. \violet{Put differently, the VC dimension lower bound may be witnessed by an extremely complicated shattering set, one that is intractable to sample from and hence arguably will not arise in the real world---indeed, this intuition is the starting point for our proof of~\Cref{thm:statistical intro}.}

\paragraph{Computational separation.} The concept class in~\Cref{thm:statistical intro} has exponential circuit complexity, necessarily so since the class of size-$s$ circuits has VC dimension $\tilde{O}(s)$.  It is natural to then ask about samplable vs.~standard PAC learning of classes of polynomial-size circuits. Such classes are of interest because \violet{functions that arise in practice can be assumed to be efficiently computable, and relatedly,} the barriers to their efficient learnability are solely computational in nature and cannot rely on information-theoretic impossibility. 

 As already observed in~\cite{Val84}, if $\mathsf{RP} = \mathsf{NP}$ then every concept class of polynomial-size circuits is efficiently learnable in polynomial time in standard PAC. So short of proving $\mathsf{RP}\ne \mathsf{NP}$, any computational separation will have to either rely on complexity assumptions or be relativized. Implicit in the work of Xiao~\cite{Xia10} is a separation relative to a {\sl specific} oracle. We discuss~\cite{Xia10}'s result in~\Cref{sec:related}, mentioning for now that this is an oracle relative to which one-way functions do {\sl not} exist. This therefore should not be viewed as evidence as to whether such a separation exists in the unrelativized world---and~\cite{Xia10} did not claim it as such---since presumably we believe that one-way functions do  exist in the unrelativized world.

Our second result gives a computational separation conditioned on two complexity assumptions, a standard one (ironically, the existence of one-way functions) and a new one that we introduce (the existence of explicit ``evasive sets",~\Cref{key assumption intro}):

\smallskip 

 \begin{tcolorbox}[colback = white,arc=1mm, boxrule=0.25mm]
\begin{theorem}[See \Cref{thm:computational-body} for the formal version]
\label{thm:computational intro}
    Assume the existence of one-way functions and explicit evasive sets. There is a concept class of polynomial-size circuits that requires superpolynomial time to learn in standard PAC, but is learnable in polynomial time in samplable PAC. 
\end{theorem}
\end{tcolorbox}\medskip

The definition of an explicit evasive set is rather technical and we defer it to~\Cref{sec:technical overview}.  For now, we mention that it is a set $H\sse \zo^n$ that is {\sl explicit} in the sense that membership in $H$ can be easily verified (i.e.~the function $x \mapsto \Ind[x\in H]$ is computable in polynomial time), and yet is {\sl evasive} in the sense that any samplable distribution must ``mostly miss" it. The crux of the connection to learning lies in pinning down the appropriate notion of ``mostly miss" (\Cref{def:epsk-miss intro}).  

\violet{

\paragraph{Random oracle separation.} Proving the existence of explicit evasive sets unconditionally is likely difficult: Like in the case of one-way functions, doing so will imply $\mathsf{P}\ne \mathsf{NP}$ (\Cref{claim:hard-to-prove}). We nevertheless prove that they exist relative to a {\sl random} oracle (\Cref{thm:hard-to-sample-relative-to-random-oracle}). Since one-way functions also exist relative to a random oracle, we obtain as a corollary a computational separation of samplable PAC from standard PAC that holds relative to a random oracle, improving on~\cite{Xia10}'s separation for a specific oracle: 

\begin{corollary}[See \Cref{thm:computational-sep-with-random-oracle} for the formal version]
\label{cor:random oracle separation}
The assumptions of~\Cref{thm:computational intro}, and hence the separation between samplable and standard PAC, hold relative to a random oracle. 
\end{corollary}

As is standard in complexity theory, we view~\Cref{cor:random oracle separation} as saying that the speedups offered by samplable PAC over standard PAC hold not just for certain structured instances (that may have been tailored for such a separation), but even for generic, unstructured ones. See~\cite{AA14} for a discussion of this point and the role of random oracle separations more generally. 
}

\begin{remark}[Evasive sets and uniform generation]
    The existence of explicit evasive sets is related to, but differs from, the hardness of {\sl uniform generation}~\cite{JVV86}. In uniform generation an algorithm is given the description of a circuit $C : \zo^n\to \zo$ and is asked to sample $\mathrm{Unif}(C^{-1}(1))$ (either exactly or approximately).  The existence of explicit evasive sets  implies the hardness of uniform generation but not vice versa: It could be the case that every polynomial-size circuit $C$ {\sl does} have a corresponding polynomial-size circuit that generates $\mathrm{Unif}(C^{-1}(1))$, but such generators are just hard to {\sl construct} efficiently. 
    \end{remark}



\subsection{Extensions}


\paragraph{Separations within samplable PAC.} The  techniques we use to prove~\Cref{thm:statistical intro,thm:computational intro} extend to give finer-grained separations {\sl within} samplable PAC, showing that learning under distributions with size-$s$ generators can be much easier than under those with size-$S$ generators, even if $s$ is only slightly smaller than $S$.

\begin{theorem}[See \Cref{thm:separations within samplable PAC formal} for the formal version]
    \label{thm:separations within samplable PAC intro}
     For every $s \ge n$ there is a concept class that is learnable with polynomial sample complexity under distributions with size-$s$ generators, and yet requires exponential sample complexity under those with size-$S$ generators for $S \ge \Omega(sn \log s)$.  
\end{theorem}

While~\Cref{thm:statistical intro} shows that there are learning tasks whose sample complexity scale smoothly with the complexity of the  distribution,~\Cref{thm:separations within samplable PAC intro} shows that there are ones for which a slight increase in the complexity of the distribution results in a dramatic increase in sample complexity. See~\Cref{fig:graph2}. We also prove a computational analogue of~\Cref{thm:separations within samplable PAC intro}. See~\Cref{thm:computational analogue of separations within samplable PAC formal}.

\begin{figure}[h]
\hspace{-1.2cm}
        \begin{tikzpicture}
         \pgfplotsset{small}
        \matrix {
            \begin{axis}[
                        scale only axis,
                        width=0.25\linewidth,
                        xtick={0,0.5,1},
                        xticklabels={$1$,{\small $\poly(n)$},$2^n$},
                        ytick={0,0.5,1},
                        yticklabels={$1$,{\small $\poly(n)$},$2^n$},
                        xmin=0, xmax=1,
                        ymin=0, ymax=1,
                        xlabel = {Distribution complexity},
                        ylabel = {Sample complexity},
                        y label style = {text width=1.8cm,align=left, rotate=-90, at={(-0.4,0.5)}},
                        axis lines=left,
                        clip=false,
                        xtick style={draw=none},
                        ytick style={draw=none},
                    ]
                \addplot[color=blue,smooth,thick,-,domain=0:1, forget plot] {x};
                \addplot[gray, dashed, forget plot] coordinates {(0.5,0) (0.5,0.5)};
                \addplot[gray, dashed, forget plot] coordinates {(0,0.5) (0.5,0.5)};
            \end{axis}
            & \hspace{1cm}
            \begin{axis}[
                        scale only axis,
                        width=0.25\linewidth,
                        xtick={0,0.25,0.5,1},
                        xticklabels={$1$,{},{\small $S$},$2^n$},
                        ytick={0,0.25,1},
                        yticklabels={$1$,$\poly(s)$,$2^n$},
                        xmin=0, xmax=1,
                        ymin=0, ymax=1,
                        xlabel = {Distribution complexity},
                        ylabel = {Sample complexity},
                        y label style = {text width=1.8cm,align=left, rotate=-90, at={(-0.3,0.5)}},
                        axis lines=left,
                        clip=false,
                        xtick style={draw=none},
                        ytick style={draw=none}
                    ]
                \addplot[color=teal,smooth,thick,-,domain=0:0.5, forget plot] {x};
                \addplot[color=teal,smooth,thick,-,domain=0.5:1, forget plot] {1};
                \addplot[gray, dashed, forget plot] coordinates {(0.5,0) (0.5,1)};
                \addplot[gray, dashed, forget plot] coordinates {(0.25,0) (0.25,0.25)};
                \addplot[gray, dashed, forget plot] coordinates {(0,0.25) (0.25,0.25)};
                \filldraw[black] (.25,-.106) node[] {\small $s$};
            \end{axis} \\
            };
        \end{tikzpicture}

        \captionsetup{width=.93\linewidth}
    \caption{The left and right plots illustrate how the sample complexities of the learning tasks in~\Cref{thm:statistical intro,thm:separations within samplable PAC intro} respectively scale with the complexity of the distribution. See their formal versions for the quantitative parameters. 
    } 
    \label{fig:graph2}
\end{figure}
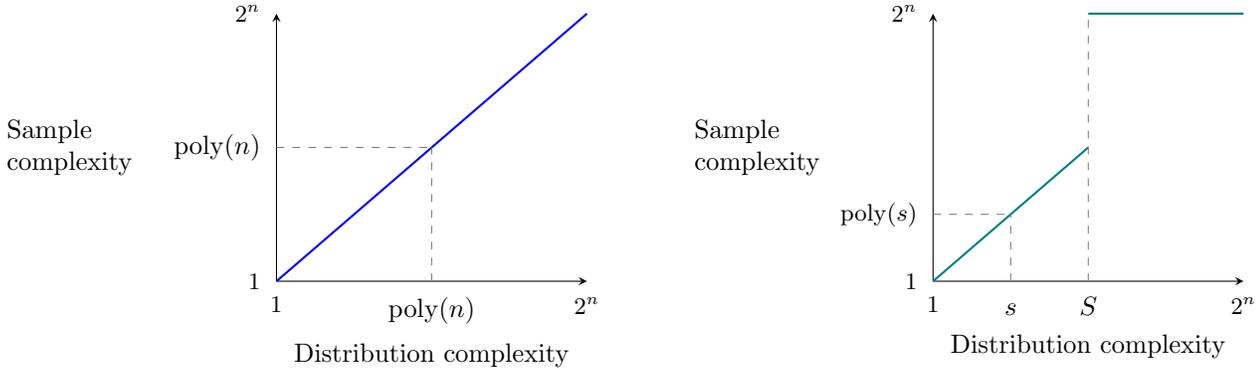



\paragraph{Online learning.} Another well-studied model of supervised learning is the online mistake bound model~\cite{Lit88}. Learning in this model proceeds in rounds. In each round, the adversary presents the learner with an unlabeled instance $x \in \zo^n$. The learner responds with its prediction $y \in \zo$ and is then told whether that is correct (i.e.~whether $y = f(x)$).  The goal of the learner is to minimize its total number of mistakes. 

Here again the standard definition allows the adversary to be computationally unbounded---it can take superpolynomial time to produce the test instance in each round. Yet, efficient online learners are expected to be efficient even against such adversaries. For the same reasons as in the PAC setting, it is therefore natural to consider the variant where the adversary is also assumed to be efficient. In~\Cref{sec:online} we show how our techniques can be extended to the online setting to similarly show how the complexity of online learning---both in terms of mistake bounds and the runtime of learners---can depend on the power of the adversary.




\section{Technical overview}
\label{sec:technical overview}

 Our proofs center around a new notion of {\sl evasive sets} that we introduce and study. These are sets that, as their name suggests, evade all samplable distributions. Let us now make this precise.


\subsection{Defining evasive sets}  

Given a distribution $\mathcal{D}$ and a  set $H$, both over $\zo^n$, we say that $\mathcal{D}$ {\sl $\eps$-misses} $H$ if it places less than $\eps$ mass on $H$: 

\begin{definition}[$\eps$-miss]
    \label{def:eps miss}
    A distribution $\mcD$ \emph{$\eps$-misses} a set $H$ if $\mcD(H) \coloneqq \Prx_{\bx \sim \mcD}[\bx \in H] < \eps$. Otherwise, we say that $\mathcal{D}$ \emph{$\eps$-hits} $H$.
\end{definition}    

A first attempt at the definition of an evasive set $H$ is one for which all samplable distributions $\eps$-miss it. However, no such $H$ can exist. For any $H$, a size-$O(n)$ circuit can memorize a specific $x \in H$ and generate the distribution that places all its mass on $x$.  This distribution $1$-hits $H$.   More generally, a size-$s$ circuit can memorize $\approx s/n$ many points in $H$. We therefore modify~\Cref{def:eps miss} to exclude the heaviest elements of the distribution: 

\medskip  

 \begin{tcolorbox}[colback = white,arc=1mm, boxrule=0.25mm]
\begin{restatable}[$(\eps,k)$-miss]{definition}{epskmiss}\label{def:epsk-miss intro}
     A distribution $\mathcal{D}$ \emph{$(\eps,k)$-misses $H$} if there exists a set $H^*$ of size $k$ such that $\mathcal{D}(H\setminus H^*) < \eps$. Otherwise, we say that $\mathcal{D}$ \emph{$(\eps,k)$-hits} $H$. 
\end{restatable}
\end{tcolorbox}\medskip

See~\Cref{fig:eps-k-missing} for an example that illustrates this definition. 

\begin{figure}[!htb]
    \centering

\begin{tikzpicture}

\draw[blue, thick] (3,3) rectangle (9,9);
\node[blue] at (9.3,9.3) {\Large$H$};

\draw[magenta, thick, fill=magenta!5] (3.3,7.5) rectangle (6.8,8.5);
\draw[<-, magenta, thick] (3.3,8.0) -- (2.5,8.0);
\node[magenta, left] at (2.5,8.0) {\footnotesize$\varepsilon$ weight};

\draw[magenta, thick, fill=magenta!15] (0,4.5) rectangle (2,7.5);
\draw[<-, magenta, thick] (1.0,4.5) -- (1.0,3.5);
\node[magenta, below] at (1.0,3.5) {\footnotesize$0.1 - \varepsilon$ weight};

\draw[magenta, thick, fill=magenta!70] (6.9,4.9) rectangle (7.6,5.6);
\draw[<-, magenta, thick] (7.25,4.9) -- (7.25,4.3);
\node[magenta, below, align=center] at (7.2,4.3) {\footnotesize$0.9$ weight on $k$ points};
\node[magenta] at (8,5.8) {\large$H^*$};

\node[magenta] at (1.0,6.0) {\Large$\mathcal{D}$};

\end{tikzpicture}

            \captionsetup{width=.93\linewidth}
    \caption{The weight distribution of $\mathcal{D}$ is illustrated by the 3 pink rectangles. The sizes of these rectangles depict the number of points and their shades depict the amount of weight. Since $\mathcal{D}$ places $0.9+\eps$ weight on $H$, it $(0.9+\eps)$-hits $H$. However, since $0.9$ amount of this weight is concentrated on the $k$ points in $H^*$, it $(\eps',k)$-misses $H$ for any $\eps' > \eps$.} 
    \label{fig:eps-k-missing}
\end{figure}
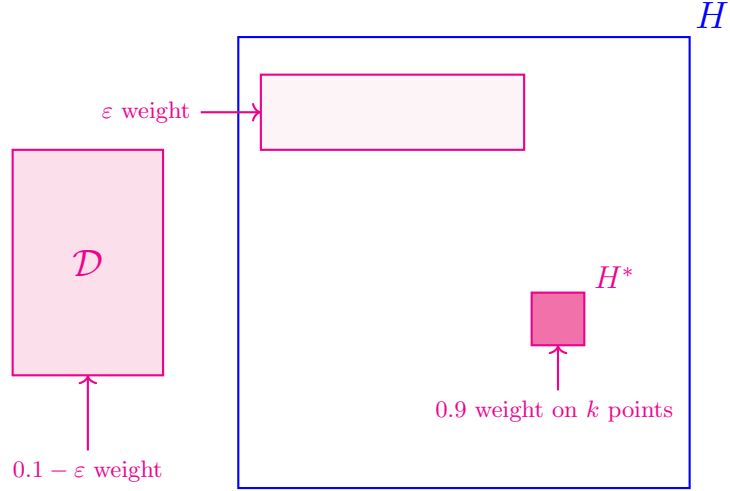

\begin{remark}[Comparison with TV distance]
\label{rem:miss vs TV}
This notion is stronger than $\mathcal{D}$ having large TV distance from $\mathrm{Unif}(H)$. If $\mathcal{D}$ $(\eps,k)$-misses $H$ then
    \[ \dist_{\mathrm{TV}}(\mathcal{D},\mathrm{Unif}(H)) \ge 1-\eps - \frac{k}{|H|}.\]
On the other hand, distributions with large TV distance from $\mathrm{Unif}(H)$ can hit $H$. For example, for any set $H$ and a distribution~$\mcD$ that is uniform on $Ck$ points within $H$, the TV distance between~$\mathcal{D}$ and $\mathrm{Unif}(H)$ is large, $1-\frac{Ck}{|H|}$,  and yet $\mcD$ $(1-\frac1{C}, k)$-hits $H$.
\end{remark}
    
We are now ready to define evasive sets. For brevity, we refer to distributions that have a size-$s$ generator as ``size-$s$ distributions". Samplable distributions are therefore size-$\poly(n)$ distributions. 

\begin{restatable}[$(\eps,k)$-evades size-$s$ distributions]{definition}{epskevades} 
     A set $H \sse \zo^n$ \emph{$(\eps,k)$-evades size-$s$ distributions} if every size-$s$ distribution $(\eps,k)$-misses $H$.
\end{restatable}

We will be interested in the regime where $\eps$ is small and $k \approx s$,  capturing the notion that  the best thing a size-$s$ distribution can do in terms of approximating $\mathrm{Unif}(H)$ is to simply memorize as many points in $H$ as its size allows and output the uniform distribution over those points. If $|H|\gg s$, as will be the case in our constructions, this is a very bad approximation of $\mathrm{Unif}(H)$.

\subsection{A conjecture about explicit evasive sets}
\label{sec:evasive}

\paragraph{A non-explicit construction.} Our statistical separation of samplable PAC from standard PAC (\Cref{thm:statistical intro}) relies on the existence of large evasive sets. Largeness will be useful for our lower bounds against standard PAC whereas the evasiveness  will be useful for our upper bounds in samplable PAC. We prove:  

\begin{restatable}[Existence of an evasive set]{lemma}{CountingArgument}
    \label{lem:non-explicit intro}
    For any $\delta \ge 2^{-n}$ there is a $\delta$-dense set $H \subseteq \zo^n$  that $(\eps, O((s\log s)/\eps))$-evades all size-$s$ distributions for all $s \ge n$ and $\eps \ge 4\delta$.
\end{restatable}

We prove~\Cref{lem:non-explicit intro} using the probabilistic method. For intuition, consider the special case of flat distributions (those that are uniform over their support). For such distributions $\mathcal{D}$, if $\supp(\mathcal{D})$ is sufficiently large, we can show that $\mcD$ is highly unlikely to $\eps$-hit a randomly chosen $\bH$. The failure probability decays exponentially with $|\supp(\mathcal{D})|$, allowing us to union bound over all size-$s$ distributions with sufficiently large supports. This argument no longer works if  $\supp(\mathcal{D})$ is small, but in this case we can use the fact that $\mcD$ trivially $(0,|\supp(\mathcal{D})|)$-misses every $H$. The parameters of~\Cref{lem:non-explicit intro} are near-optimal, since every $H$ is $(1-\lambda, \Theta(\lambda s/n))$-hit by any size-$s$ distribution that is uniform on $\Theta(s/n)$ many memorized points in $H$.





\paragraph{Explicitness.} Our computational separation of samplable PAC from standard PAC (\Cref{thm:computational intro}) relies on the existence of a large set that not only evasive, but   is furthermore {\sl explicit} in the sense that membership in it (i.e.~the function $x \mapsto \Ind[x\in H]$) is easy to decide:   

\medskip  

 \begin{tcolorbox}[colback = white,arc=1mm, boxrule=0.25mm]
\begin{conjecture}[See~\Cref{conjecture:hard-to-sample} for the formal version]
    \label{key assumption intro}
   There is a set $H \sse \zo^n$ satisfying: 
    \begin{enumerate}
        \item {\sl Explicit:} Membership in $H$ is computed by a polynomial-size circuit. 
        \item {\sl Large:} $H$ has superpolynomial size.
        \item {\sl Evasive:} $H$ $(\eps,k)$-evades all size-$s$ distributions for all $s = \poly(n)$ and $\eps = 1/\poly(n)$, where $k \le \poly(s,1/\eps)$. 
    \end{enumerate}
\end{conjecture}
\end{tcolorbox}\medskip

We study~\Cref{key assumption intro} in detail in~\Cref{sec:conjecture}. We show that it implies $\mathsf{P}\ne \mathsf{NP}$ (\Cref{claim:hard-to-prove}) and that it holds relative to a random oracle (\Cref{thm:hard-to-sample-relative-to-random-oracle}).  The proof of the latter  strengthens that of~\Cref{lem:non-explicit intro}. The key idea is to show that a samplable distribution $\mathcal{D}$ is highly unlikely to hit a randomly chosen $\bH$, even if the circuit generating $\mathcal{D}$ is allowed unit-time membership queries to $\bH$.

\subsection{The connection to PAC learning}
\label{subsec:overview}


\paragraph{Proof \red{overview} of~\Cref{thm:statistical intro}.} We first describe how~\Cref{lem:non-explicit intro} yields a statistical separation of samplable PAC from standard PAC. For a set $H \sse \zo^n$ and function $f : \zo^n\to \zo$, we write $f_H : \zo^n\to\zo$ to denote the following restriction of $f$ to $H$: 
\[ f_H(x) = 
\begin{cases}
    f(x) & \text{if $x\in H$} \\ 
    0 & \text{otherwise.}
\end{cases}
\]
For a concept class $\mathcal{C}$, we similarly write $\mathcal{C}_H$ to denote the restriction of $\mathcal{C}$ to $H$: 
\begin{equation}
    \label{eq:concept-class-restriction}
    \mathcal{C}_H = \{ f_H \colon f \in \mathcal{C}\}.
\end{equation}
Now consider $\mathcal{A}_H$ where $\mathcal{A}$ is the class of all functions $f : \zo^n\to \zo$. It is easy to check that~$H$ is the largest set shattered by $\mathcal{A}_H$ and hence the VC dimension of $\mathcal{A}_H$ is exactly $|H|$. The sample complexity of learning $\mathcal{A}_H$ in standard PAC is therefore governed by the size of $H$. In particular, if~$H$ has exponential size then learning $\mathcal{A}_H$ in standard PAC requires exponential sample complexity. (This is why we are concerned with evasive sets of large size in~\Cref{lem:non-explicit intro}.)

On the other hand, we are able to exploit the evasiveness of $H$ to design an efficient algorithm for learning $\mathcal{A}_H$ in samplable PAC: 

\begin{restatable}[Evasiveness implies efficient learners]{lemma}{LearnEvasiveH}
\label{lem:evasiveness implies learnability intro} 
Let $H \sse \zo^n$ be a set that $(\eps,k)$-evades size-$s$ distributions. Then for any concept class $\mathcal{C}$ there is an algorithm for learning $\mathcal{C}_H$ to error $O(\eps)$ under all size-$s$ distributions using $O(k/\eps)$ samples and running in time $O(nk/\eps)$. 
\end{restatable}

The intuition for~\Cref{lem:evasiveness implies learnability intro} is simple. If $H$ is $(\eps, k)$-evasive, then a learner that memorizes the labels for the $k$ heaviest points in $H$ will only incur  $\eps$ error. While a learner may not see the $k$ heaviest points, or even know when it has seen them, we show that memorizing $O(k/\eps)$ samples suffices to achieve good accuracy.

\paragraph{Proof \red{overview} of~\Cref{thm:computational intro}.}  In the proof of~\Cref{thm:statistical intro}, since $H$ is not explicit and $\mathcal{A}$ is the class of all functions, there are no nontrivial upper bounds on circuit complexity of the functions in~$\mathcal{A}_H$. We now describe how we extend the proof strategy so that the separating concept class is a class of polynomial-size circuits. As mentioned, the lower bounds against standard PAC will now be computational in nature and can no longer rely on the information-theoretic arguments that underlie VC dimension lower bounds.

Consider the class $\mathcal{F}_H$ where $H$ is an explicit evasive set given by~\Cref{key assumption intro} and $\mathcal{F}$ is a pseudorandom function family~\cite{GGM86}, the existence of which follows from the existence of one-way functions.  First note that $\mathcal{F}_H$ now does in fact have polynomial circuit complexity: every function in this class can be computed by a circuit of size $O(S_1 + S_2)$, where $S_1$ is the circuit complexity of deciding membership in $H$ and $S_2$ is the circuit complexity of $\mathcal{F}$. 

It is well known that pseudorandom function families are hard to learn: an efficient learner for $\mathcal{F}$ in standard PAC can be used to break  $\mathcal{F}$'s security guarantees~\cite{Val84}. We extend this to show that as long as $H$ is sufficiently large, an efficient learner for $\mathcal{F}_H$ suffices to break $\mathcal{F}$'s security guarantees. See~\Cref{lem:PRFF-to-hardnes-of-learning}. Stated in the contrapositive, $\mathcal{F}$'s security guarantees implies hardness of learning $\mathcal{F}_H$ in standard PAC. (This is why we are concerned with evasive sets of large size in~\Cref{key assumption intro}.)


\section{Related work}
\label{sec:related}

\paragraph{Xiao's separation.} As mentioned, a separation between samplable PAC and standard PAC for a {\sl specific} oracle is implicit in the work of Xiao:

\begin{theorem}[Follows from Theorem 1.3 of~\cite{Xia10}]
\label{thm:xiao}
    There is an oracle $\mathcal{O}$ such that: 
    \begin{enumerate}
        \item There is a polynomial-time algorithm $A$ such that $A^\mathcal{O}$ learns $\mathsf{SIZE}^\mathcal{O}(n^2)$ in samplable PAC. 
        \item Any algorithm $A$ such that $A^\mathcal{O}$ learns $\mathsf{SIZE}^\mathcal{O}(n^2)$ in standard PAC must take superpolynomial time. 
    \end{enumerate}
\end{theorem}

An inspection of~\cite{Xia10}'s proof shows that $\mcO$ is an oracle relative to which one-way functions do not exist.\footnote{We sketch the justification here. As stated in Theorem 1.3 of~\cite{Xia10}, this is an oracle relative to which the learning of all distributions with polynomial-size generators, in the sense of Kearns, Mansour, Ron, Rubinfeld, Schapire, and Sellie~\cite{KMRRSS94}, is easy.  However, as shown in~\cite{KMRRSS94}, the hardness of  this task is implied by the existence of one-way functions.  Since this implication relativizes,~\cite{Xia10}'s oracle is one relative to which one-way functions do not exist.} Since we believe that one-way functions exist in the unrelativized world, this therefore does not shed much light on the relationship between samplable and standard PAC in the unrelativized world.\footnote{\violet{As in the case of~\cite{BFKL93}, the focus of~\cite{Xia10} was not on the distinction between standard and samplable PAC. Rather, the author had proved a result that only held for samplable PAC, and he obtained~\Cref{thm:xiao} on route to showing that an extension of his result to standard PAC will require nonrelativizing techniques. 
Similarly, see also~\cite{HN22}, where a separation is given relative to an oracle for which every problem in {\sf PH} is easy on average.}}  

\violet{There are several works giving learning algorithms that work in samplable PAC, and its relaxations, but are not known to work in standard PAC (see e.g.~\cite{BCKRS02,HN22,GK23,Kar24} for a few recent examples).}

\paragraph{Distribution-specific learning.}  Given the apparent difficulty of designing efficient algorithms in standard PAC, there has been a large body of work on {\sl distribution-specific} learning. Here the data is promised to be drawn from a specific distribution, e.g.~the uniform distribution. The main downside is that this is a stylized assumption that limits the practical relevance of the model: We want our algorithms to succeed for as broad a class of distributions as possible, not just a specific one. In the case of the uniform distribution in particular, it does not capture much of the richness real-world distributions that stem from correlations among features.

Samplable PAC can be viewed as a middle ground that simultaneously corrects for the overly stringent requirements of standard PAC and the overly strong assumptions of distribution-specific PAC. 

\paragraph{Lifting uniform-distribution learners.} In the same spirit of bridging this gap between standard PAC and distribution-specific PAC, recent works~\cite{BLMT23,BLST25} show how uniform-distribution PAC learners can be generically ``lifted" to also succeed under various non-uniform yet still structured classes of distributions. 

Compared to these works, our work attempts to bridge the gap ``from the opposite direction". While these lifters scale {\sl up} the distribution-specific model, the samplable PAC model scales {\sl down} the distribution-free model (i.e.~standard PAC). 

\paragraph{Computable PAC and online learning.} Another recent line of work~\cite{AABDLU20,AABDLU21,Ste22,DRKRS23}  studies the distinction between standard PAC and a variant known as {\sl computable} PAC  where learners are restricted to be computable. Among other results, these works show that there are classes with finite VC dimension that are not learnable in computable PAC. See also~\cite{HBD23,DRKS25} for the online analogue.

The focus of our work is on statistical and computational complexity in settings where computability is not an issue, rather than the distinction between computable and uncomputable learners. 

 \paragraph{Samplable distributions in average-case complexity.}  Outside of learning theory, samplable distributions are central to the study of average-case complexity~\cite{Lev86,BDCG89}. While $\mathsf{P} \ne \mathsf{NP}$ rules out the possibility of efficient algorithms that solve $\mathsf{NP}$-hard problems on  all instances, average-case complexity is concerned with the possibility of efficient algorithms that   solve most instances generated by a samplable distribution (i.e.~the possibility that $\NP$ is hard in the worst case but easy on average).   As in samplable PAC, in this context samplable distributions are taken as the formalization of distributions that actually occur in practice.

\section{Discussion and future work}
\violet{
\begin{figure}[t]
\centering
\begin{tikzpicture}[line width=0.8pt, line cap=round, line join=round]

  \def\W{12}   
  \def\H{8}    

  \def\ex{4.6} 
  \def\ey{2.7} 

  \def\cx{6}   
  \def\cy{4}   

  \draw (0,0) rectangle (\W,\H);
  \node[anchor=north west] at (0.25,{\H-0.25}) {All distributions};

  \draw (\cx,\cy) ellipse [x radius=\ex, y radius=\ey];
  \node at (\cx,{\cy+\ey-0.63}) {Samplable distributions};

  \def\blobshift{0.3} 

  \def\rbase{1.80}
  \def\aA{0.22}
  \def\aB{0.15}
  \def\aC{0.04}
  \def\kA{6}
  \def\kB{10}
  \def\kC{14}
  \def\phiA{18}
  \def\phiB{137}
  \def\phiC{-72}

  \draw
    plot[variable=\t, domain=0:360, samples=361, smooth cycle]
      ({\cx + ( \rbase
                + \aA*sin(\kA*\t + \phiA)
                + \aB*sin(\kB*\t + \phiB)
                + \aC*sin(\kC*\t + \phiC) ) * cos(\t)},
       {\cy - \blobshift
          + ( \rbase
              + \aA*sin(\kA*\t + \phiA)
              + \aB*sin(\kB*\t + \phiB)
              + \aC*sin(\kC*\t + \phiC) ) * sin(\t)});

  \node at (\cx,{\cy - \blobshift}) {\shortstack{Real-world\vspace{2pt}\\distributions}};

\end{tikzpicture}
\caption{An illustration of how samplable distributions relate to real-world distributions.}
\label{fig:bubbles}
\end{figure}
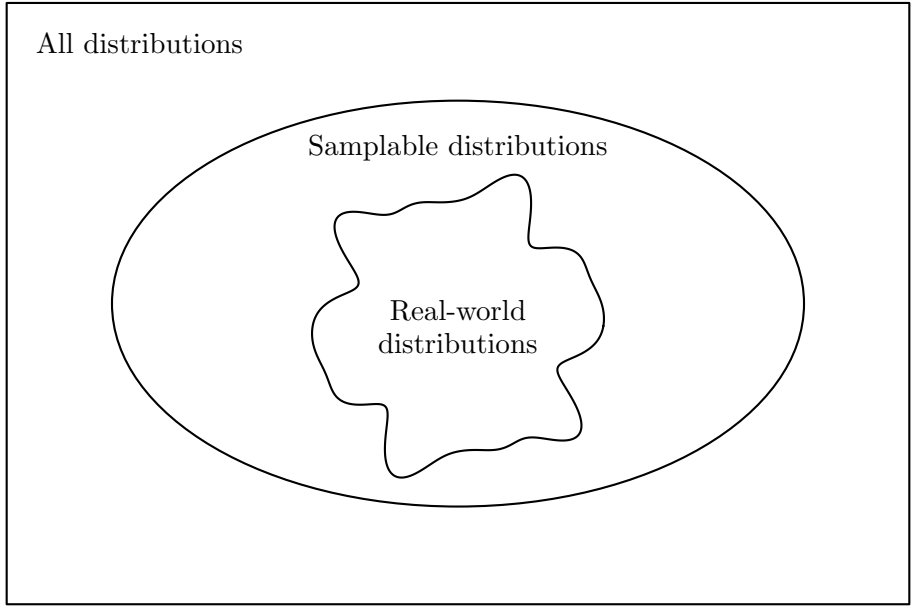
Samplability is generally viewed as a baseline requirement for real-world distributions, not a characterization. See~\Cref{fig:bubbles}. Samplable PAC is therefore only a first cut at refining the standard PAC model. 

We see two takeaways from this. First, as alluded to in the introduction, lower bounds in standard PAC---be they statistical or computational---may be overly pessimistic: While certain learning tasks may have hard instances, these hard instances themselves may be hard to find, and arguably will not occur in practice. Second, this calls for a better understanding of the actual structure of real-world distributions beyond just samplability, which can then be leveraged in the design of learning algorithms. This falls within the overall agenda of going beyond the worst-case analysis of algorithms~\cite{Rou21}. Our techniques suggest the possibility of deeper connections to the complexity of sampling and uniform generation.

A concrete open problem is that of  characterizing the sample complexity of learning in samplable PAC. Sample complexity in standard PAC is characterized by VC dimension---what is the corresponding characterization for samplable PAC? More generally, is there a characterization that takes both function and distribution complexity into account? 
}

\section{Preliminaries}

\paragraph{Basic notation and writing conventions. }
We write $[n]$ to denote the set $\{1,2,\ldots,n\}$. The length of a bitstring $x\in\zo^n$ is $|x|$. \textbf{Boldface} letters, e.g.~$\bx,\bH$, denote random variables. We write $\mathrm{Unif}(H)$ to denote the uniform distribution over the set $H$ and $\mathrm{Ber}(\delta)$ to denote the Bernoulli distribution with mean $\delta$. For $H\sse \zo^n$, the \emph{density} of $H$ is $|H|/2^n$. If $|H|/2^n \ge \delta$, we say that $H$ is $\delta$-dense. We write $\overline{H}=\zo^n\setminus H$ to denote the complement of $H$. Given a distribution $\mcD$ over $\zo^n$ and a point $x\in \zo^n$, we let $\mcD(x)\coloneqq \Prx_{\by\sim \mcD}[x=\by]$. Similarly, for a set $H\subseteq \zo^n$, we let $\mcD(H)\coloneqq \Prx_{\bx \sim \mcD}[\bx \in H]$.

We use two basic asymptotic conventions: We say a function $\sigma(n)$ is \emph{superpolynomial} if for every polynomial $p(n)$, there is some $N$ s.t.~$\sigma(n) > p(n)$ for all $n \geq N$. Similarly, we say a function $\sigma(n)$ is \emph{negligible} if $1/\sigma(n)$ is superpolynomial.

We will make use of standard concentration inequalities for sums of independent random variables.

\begin{fact}[Chernoff bound for bounded random variables]\label{fact:chernoff}
Let $\bX_1$,...,$\bX_n$ be independent random variables such that $\bX_i\in[0,1]$ for all $i$. Let $\bX=\bX_1+...+\bX_n$ denote their sum and $\mu =\Ex[\bX]$ denote that sum's expected value. Then for any $\gamma\geq 0$, 
\begin{align*}
    \Pr\big[\bX\geq (1+\gamma)\mu\big]\leq \exp\Big(-\frac{\gamma^2\mu}{2+\gamma}\Big) \quad\quad\text{and}\quad\quad \Pr\big[\bX\leq (1-\gamma)\mu\big]\leq \exp\Big(-\frac{\gamma^2\mu}{2}\Big).
\end{align*}
\end{fact}


\subsection{Evasive sets and restrictions of functions to sets}

We restate the following definitions and notational conventions from the technical overview: 

\epskmiss*

\epskevades*

For a set $H \sse \zo^n$ and function $f : \zo^n\to \zo$, we write $f_H : \zo^n\to\zo$ to denote the following restriction of $f$ to $H$: 
\[ f_H(x) = 
\begin{cases}
    f(x) & \text{if $x\in H$} \\ 
    0 & \text{otherwise.}
\end{cases}
\]
For a concept class $\mathcal{C}$, we similarly write $\mathcal{C}_H$ to denote the restriction of $\mathcal{C}$ to $H$ (i.e.~$\mathcal{C}_H = \{ f_H \colon f \in \mathcal{C}\}).$ We let $\mcA$ denote the class of all functions $f: \zo^n\to\zo$.

\subsection{Circuits and circuits generating distributions}
We consider Boolean circuits consisting of AND, OR, and NOT gates. The size of a circuit is the number of gates contained in it. By a standard counting argument, the number of circuits of size $s$ over $n$ input bits is at most $s^{O(n+s)}$. 

\begin{fact}[The number of size-$s$ circuits]
\label{fact:circuits}
    The number of Boolean circuits of size $s$ over $n$ inputs is at most $s^{O(n+s)}$. 
\end{fact}


\begin{definition}[Size-$s$ distribution]
\label{def:distribution size}
    We say a distribution $\mcD$ over $\zo^n$ is \emph{size-$s$} if there is some size-$s$ generating circuit $\mcG:\zo^\ell \to \zo^n$ for which the distribution of $\mcG(\Unif (\zo^\ell))$ is exactly equal to $\mcD$. We refer to $\mcG$ as a \emph{generator} for $\D$.
\end{definition}


\subsection{Learning theory}

\begin{definition}[PAC learning \cite{Val84}]
    \label{def:PAC-learning}
    For any concept class $\mcC$, we say an algorithm $A$ \emph{learns $\mcC$ to error $\eps$ over distribution $\mathcal{D}$ using $m$ samples} if the following holds: For any $f \in \mcC$, given $m$ independent samples of the form $(\bx, f(\bx))$ where $\bx \sim \mathcal{D}$, $A$ returns a hypothesis $h$, that with probability at least $2/3$, satisfies $\Prx_{\bx\sim \D}[f(\bx)\neq h(\bx)]\le \eps.$
    
    We say that $A$ learns $\mcC$ to error $\eps$ using $m$ samples if \emph{for every} distribution $\D$, $A$ learns $\mcC$ to error $\eps$ over the distribution $\mathcal{D}$ using $m$ samples.
\end{definition}

\begin{definition}[VC dimension \cite{VC71}]
    Let $\mcC$ be a concept class consisting of Boolean functions. A set $S\sse \zo^n$ is \emph{shattered} if for every labeling $\ell:S\to\zo$, there is an $f\in\mcC$ such that $f(x)=\ell(x)$ for every $x\in S$. The \emph{VC dimension} of $\mcC$, $\VCdim(\mcC)$, is the size of the largest set which is shattered. 
\end{definition}

We will use the following fundamental result in learning theory which constructs a distribution over a shattering set and shows that learning over this distribution requires many samples.

\begin{fact}[Hard distribution over shattering set~\cite{BEHW89,EHKV89}]
    \label{fact:shattering}
    Let $S=\{x^{(1)},\ldots,x^{(d)}\}$ be a shattering set for a concept class $\mcC$. For any $\eps\le 1/8$, let $\D$ be the distribution defined by
    $$
    \D(x)= \begin{cases}
        1-8\eps & \text{if }x=x^{(1)}\\
        8\eps / (d-1) & \text{otherwise}
    \end{cases}
    $$
    then, any algorithm which learns $\mcC$ over $\D$ must use at least $\Omega(d/\eps)$ samples. 
\end{fact}

\subsection{Limited independence generators and their properties}

In this section, we recall standard definitions of limited independence and PRGs for limited independent distributions along with some basic facts that will be helpful in our proofs.

\begin{definition}[$d$-wise independent random variables]
    We say that a collection of random variables $\{\bX^{(1)},\ldots,\bX^{(m)}\}$ is $d$-wise independent if for all $S\sse [m]$ with $|S|\le d$, the random variables $\{\bX^{(i)}\}_{i\in S}$ are independent.
\end{definition}

\begin{definition}[$\delta$-biased $d$-wise independent distributions over $\zo^N$] 
    Let $\D$ be a distribution over $\zo^N$ and let $\bx\sim \D$. We say $\D$ is $d$-wise independent if the collection of random variables $\{\bx_i\}_{i\in [N]}$ are $d$-wise independent. Furthermore, we say that the distribution is $\delta$-biased if $\bx_i\sim \mathrm{Ber}(\delta)$ for all $i\in [N]$.  
\end{definition}

\begin{definition}[Explicit PRGs for marginally uniform $d$-wise independent random variables and $\delta$-biased, $d$-wise independent distributions over bitstrings]
\label{def:prgs for limited independent rvs}
    A function $\mcG:\zo^\ell\to (\zo^n)^m$ is a pseudorandom generator (PRG) for a collection of $d$-wise independent random variables over $\zo^n$ if the random variables $\{\bx^{(i)}\}_{i\in [m]}$ obtained by sampling $\br\sim \zo^\ell$ and setting $\mcG(\br) = (\bx^{(1)},\ldots,\bx^{(m)})$ are $d$-wise independent. The random variables are marginally uniform if $\bx^{(i)}\sim\mathrm{Unif}(\zo^n)$ for all $i\in [m]$. We say that $\mcG$ is explicit if there is a circuit $C:\zo^\ell\times\zo^{\log m}\to \zo^n$ of size $\poly(\ell, \log m, n)$ such that $C(r,i) = \mcG(r)^{(i)}$ for all $r\in \zo^\ell$ and $i\in [m]$.  

    For $n=1$, we say that $\mcG$ is a $\delta$-biased, $d$-wise independent PRG if $\bx^{(i)}\sim \mathrm{Ber}(\delta)$. We say $\mcG$ is explicit if, as above, there is a circuit which computes it of size $\poly(\ell, \log m, \log (1/\delta))$.
\end{definition}

\begin{fact}[Existence of explicit PRGs for $d$-wise independent random variables and for $\delta$-biased $d$-wise independent distributions over bitstrings]
\label{fact:explicit prgs for rvs}
    For all $d$, $n$, and $m\le 2^n$, there exists an explicit $\mcG:\zo^\ell\to (\zo^n)^m$ for a collection of marginally uniform $d$-wise independent random variables over $\zo^n$ with $\ell=O(dn)$. Furthermore, the size of the circuit computing $\mcG$ is $O(\ell\log \ell)$.

    Also, for any $r\le n$, there exists an explicit $\mcG:\zo^{\ell}\to \zo^{m}$ for $\ell=O(d n)$ which is a $2^{-r}$-biased, $d$-wise independent \textnormal{PRG}. Furthermore, the size of the circuit $C$ computing $\mcG$ is $O(\ell\log \ell)$.
\end{fact}

The first PRG in \Cref{fact:explicit prgs for rvs} follows via a well-known random polynomial construction \cite{Jof74}. 

The latter construction in \Cref{fact:explicit prgs for rvs} can be derived from the first by \red{taking an $r$-wise AND:} To obtain a $2^{-r}$ biased bit from a uniform random $\bm{x}\sim\zo^n$, we can output $0$ if an only if the first $r$ bits of $\bm{x}$ are 0. It is straightforward to verify that this truncation can be performed efficiently given a circuit computing $\bm{x}$.





It is well-known that Chernoff-like concentration bounds extend to the setting of sums of $d$-wise independent random variables. We will make use of the following such bounds.

\begin{fact}[Concentration bounds for sums of $d$-wise independent random variables \cite{SSS95}]
\label{fact:SSS}
    Let $\bX_1,\ldots, \bX_n$ be $d$-wise independent random variables in the range $[0,1]$. Let $\bX = \sum_{i=1}^n\bX_i$ and $\mu= \E[\bX]$. Then, for all $\gamma>0$
    $$
    \Pr[\bX \ge \mu(1+\gamma)]\le \left(\frac{e^\gamma}{(1+\gamma)^{(1+\gamma)}}\right)^\mu
    $$
    as long as $d\ge \mu\gamma / (1 - \mu/n)$.
\end{fact}

\subsection{Pseudorandom function families}
\begin{definition}[Pseudorandom function families (PRFFs) secure against non-uniform adversaries]
    \label{def:prff}
    A pseudorandom function family $\mathcal{F}=\{f_s:\zo^{n}\to\zo\}_{s\in\zo^n}$ is a collection of functions such that the following holds.
    \begin{enumerate}
        \item There exists a polynomial-time algorithm that given $s,x\in\zo^n$ computes $f_s(x)$.
        \item For all polynomial-time oracle algorithms with advice $A$ and a negligible function $\sigma(n)$, we have
        \begin{equation*}
        \begin{split}
        \left|
            \begin{aligned}
                    &\quad\Prx_{\bs\sim \zo^n}[A(1^n)\text{ outputs }1\text{ when given oracle access to }f_{\bs} ]\\
                    &\quad \qquad -\Prx_{\boldf:\zo^n\to\zo}[A(1^n)\text{ outputs }1\text{ when given oracle access to }\boldf ]\quad
                \end{aligned}
        \right| < \sigma(n)
        \end{split}
        \end{equation*}
        where $\boldf:\zo^n\to\zo$ is a uniformly random  function.
    \end{enumerate}
\end{definition}

A classic result in cryptography shows that one way functions (OWFs) secure against non-uniform adversaries imply the existence of PRFFs secure against non-uniform adversaries \cite{GGM86}. Since we only consider non-uniform adversaries, when we write OWFs, we mean those that give rise to the PRFFs in \Cref{def:prff}.

\section{Proof of~\Cref{thm:statistical intro}}

In this section we prove~\Cref{lem:non-explicit intro,lem:evasiveness implies learnability intro} and show how together they yield~\Cref{thm:statistical intro}. Subsequent sections will build on the proof of~\Cref{lem:non-explicit intro} and also use~\Cref{lem:evasiveness implies learnability intro}.

\begin{theorem}[Formal version of~\Cref{thm:statistical intro}]\label{thm:statistical formal}
For any $\delta\geq 2^{-n}$, there is a concept class $\mcC$ over $\zo^n$ such that
\begin{enumerate}
    \item The VC dimension of $\mcC$ is $\ge \delta   2^n$ and 
    \item $\mcC$ is learnable under size-$s$ distributions to error $O(\eps)$ with sample complexity $O((s\log s)/\eps^2)$ and runtime $ O(n(s\log s)/\eps^2)$ for all $s\geq n$ and  $\eps\geq 4\delta$.
\end{enumerate}
\end{theorem}

\Cref{thm:statistical intro} follows immediately from \Cref{thm:statistical formal} by choosing $\delta=2^{-n/2}$. With this $\delta$, we have that $\mcC$ is learnable over size-$s$ distributions to error $O(\eps)$ using $\poly(n,s,1/\eps)$ samples and runtime. Yet~$\mcC$ requires $2^{\Omega(n)}$ samples to learn in standard PAC even just to error $0.1$.



\subsection{Proof of~\Cref{lem:non-explicit intro}}

In this section, we will find an $H$ that evades all samplable distributions. This $H$ will form the basis for the concept class $\mcC_H$ that will be learnable in samplable PAC. 
 

We prove~\Cref{lem:non-explicit intro} using the probabilistic method. We first show that any {\sl fixed} distribution, regardless of samplability, hits a random $\bH$ with small probability (\Cref{claim:fully-random misses}). Taking a union bound over all samplable distributions, we derive the existence of an $H$ that evades all of them. 


\begin{claim}
[Distributions are unlikely to hit  $\bH$]\label{claim:fully-random misses}
 Let $\delta>0$, $\eps\geq 2\delta$, and $k\in \mathds{N}$. Let $\bH\sse\zo^n$ be a random subset such that each $x\in\zo^n$ is included in $\bH$ independently with probability $\delta$. Then for any distribution $\mcD$, 
\begin{align*}
        \Prx_{\bH}[\text{$\mcD$ \emph{$(\eps, k)$-hits} $\bH$}]\leq e^{-\eps k/6}.
\end{align*}
\end{claim}

\begin{proof}
Let $T\subseteq \zo^n$ consist of $x$ such that $\mcD(x)\geq \frac{1}{k}$. Notice that $|T|\leq k$. So if $\mcD$ $(\eps,k)$ hits $H$, then by definition $\mcD(H\setminus T)\geq \eps$. Therefore, 
\begin{align*}
    \Prx_{\bH}[\text{$\mcD$ \emph{$(\eps, k)$-\emph{hits}} $\bH$}]\leq \Prx_{\bH}[\mcD(\bH\setminus T)\geq \eps]
\end{align*}
so it suffices to bound this second probability. Unpacking this notation, we see that 
\begin{align*}
    \mcD(\bH \setminus T) = \sum_{x} \mcD(x)\cdot  \Ind[x\in \bH\setminus T]
\end{align*}
is the sum of independent random variables. Let's name these variables $\bZ_x= \mcD(x)\cdot \mathds{1}[x\in\bH\setminus T]$. By the definition of $T$, each variable is bounded $\bZ_x\in[0,1/k]$, and we can also  bound the expected value of their sum 
\begin{align*}
    \Ex_{\bH}[\mcD(\bH \setminus T)] =\Ex_{\bH}\bigg[\sum_{x} \bZ_x\bigg] \leq \sum_{x} \mcD(x) \Prx_{\bH} [x\in\bH\setminus T]\leq \delta
\end{align*}
 by the density of $\bH$. Rescaling, $k\bZ_x$ are random variables bounded in $[0,1]$ with expected sum at most $\delta k$. We can therefore apply the Chernoff bound given in~\Cref{fact:chernoff}: 
\begin{align*}
    \Prx_{\bH}[\mcD(\bH\setminus T)\geq \eps] &= \Prx_{\bH}\bigg[\sum_x k\bZ_x \geq k \eps\bigg]\\ 
    &= \Prx_{\bH}\bigg[\sum_x k\bZ_x \geq (1+(\eps-\delta)/\delta) \delta k\bigg]\\ 
    &\leq \exp\bigg(-\frac{(\eps -\delta)^2}{\delta^2(2+(\eps -\delta)/\delta)}\delta k\bigg)
    \tag{\Cref{fact:chernoff} with $\gamma =\frac{\eps-\delta}{\delta}$}\\ 
    &= \exp\bigg(-\frac{(\eps -\delta)^2}{\eps+\delta} k\bigg)\\
    &\leq \exp\bigg(-\frac{(\eps/2)^2}{3\eps/2} k\bigg)\tag{$\delta\leq \eps/2$}\\
    &=e^{-\eps k/6},
\end{align*} 
which proves the claim.  \end{proof}

We prove the following quick proposition that will allow us to give a bound that holds over all (infinitely many) values of $\eps$ by union bounding over a finite \red{set of values}. 
\newcommand{\Bad}{\mathrm{Bad}}
\begin{proposition}[Discretizing $\eps$]\label{prop:discrete_eps}
Let $\eps\geq 2^{-n}$, $k\in \mathds{N}$, and $H\subseteq \zo^n$. 
\red{Let $\mcD$ be a distribution that $(\eps, O(k/\eps))$-hits $H$.}
Then there exists $i\in[n]$ such that $\mcD$ also $(2^{-i}, O(k/ 2^{-i}))$-hits $H$.
\end{proposition}
\begin{proof}
    Let $i$ be the smallest integer such that $2^{-i}\leq \eps$. Since $2^{-i} \leq \eps \leq 2^{-i+1}$, we have $O(k/\eps)= O(k/2^{-i})$;
    \red{thus $\mcD$ $(\eps, O(k/2^{-i}))$-hits $H$.
    We observe that if $\mcD$ $(\eps, O(k/2^{-i}))$-hits $H$, then $\mcD$ also $(\eps', O(k/2^{-i}))$-hits $H$ for any $\eps' \le \eps$.
    Thus, $\mcD$ $(2^{-i}, O(k/ 2^{-i}))$-hits $H$.}
    
\end{proof}

We are now ready to prove~\Cref{lem:non-explicit intro}. We restate it here for convenience.

\CountingArgument*

\begin{proof}
\red{Let $\delta' = 2\delta$ and $\bH$ be a $\delta'$-biased random set.} Let $\Bad(s,\eps,k,\bH)$ be the event that there exists a size-$s$ distribution that $(\eps,k)$-hits $\bH$.  We will show that, for a random $\bH$, the probability that $\Bad(s, \eps, O((s \log s)/\eps), \bH)$ occurs for any $s \geq n$ and \red{$\eps \geq 2\delta'$} is strictly less than $1$, implying such an $H$ exists.

    \red{By~\Cref{fact:circuits}, there exists a constant $c_1$ such that there are at most $s^{c_1(n+s)}$ circuits of size $s$.}
    Then, applying~\Cref{claim:fully-random misses} and a union bound over all size-$s$ distributions, we have
    \red{
      \begin{align*}
        \Prx_{\bH}[\Bad(s,\eps,k,\bH)]\leq e^{-\eps k/6}e^{c_1(s+n)\log s}.
    \end{align*}
    We then choose $k^\star \coloneqq 6(c_1+c_2)(s+n)
\log(s)/\eps)$, for a sufficiently large constant $c_2$:
\begin{align*}   \Prx_{\bH}\Big[\Bad\Big(s,\eps,k^\star,\bH\Big)\Big]\leq e^{-c_2(s+n)\log s}.       
\end{align*}
Taking a union bound over all possible values of $s$,
\begin{align*}
     \Prx_{\bH}\Big[\exists \text{ $s\in \mathds{N}$ such that $\Bad\Big(s,\eps,k^\star,\bH\Big)$}\Big]\leq \sum_{s=1}^\infty e^{-c_2(s+n)\log s}.
\end{align*}
    }
In order to prove that we can have a single choice of $H$ for all \red{$\eps \geq 2\delta' = 4\delta$}, we also need to union bound over all values of $\eps$. 
Since $\delta\geq 2^{-n}$, we have that \red{$\eps\geq 2\delta' > 2^{-n}$}. Thus, by~\Cref{prop:discrete_eps}, it suffices to only consider discretized $\eps$ of the form $2^{-i}$ for $i\in [n]$.
Then, union bounding over these $n$ values of $\eps$ we have, 
\red{
\begin{align*}
     \Prx_{\bH}\Big[\exists  \text{ $s\in \mathds{N}$ and } \eps &\geq 2\delta' \text{ s.t. $\Bad\Big(s,\eps,k^\star,\bH\Big)$}\Big]\\
     &\leq \sum_{s=1}^\infty n\cdot e^{-c_2(s+n)\log s}\\
     &= \sum_{s=1}^\infty e^{-c_2(s+n)\log s + \ln n}.
\end{align*}
}

The above is a convergent series, so by choosing \red{$c_2$} appropriately, it is less than some small constant---$0.1$ will suffice for our application. We have therefore shown that 0.1 of the randomly chosen $H$ do indeed $(\eps, O((s\log s)/\eps))$-evade all size-$s$ distributions for all $s\geq n$ and \red{$\eps\geq 2\delta'$}. 

It remains to further show that there exists such an $H$ of $\delta$ density. To prove this, it suffices to show that the probability that $\bH$ has less than $\delta$ density is smaller than $0.9$. As we just argued, the probability that $\bH$ is not evasive is less than $0.1$, so another union bound would complete the proof. Therefore, for the remainder of the proof, we bound the probability that $\bH$ is not as dense as we'd like. We do this via a simple Chernoff bound
\red{
\[ 
    \Prx_{\bH} \big[|\bH| \le \delta 2^n] =\Prx_{\bH} \big[|\bH| \leq \delta' 2^n/2\big] = \Prx_{\bH} \Big[\sum_x \Ind[x\in\bH] \leq \delta' 2^n/2\Big]
    \leq e^{-\delta 2^n/8}.
\] }
Noting that \red{$\delta' \cdot 2^{-n} \geq 1$}, the above probability at most $e^{-1/8} \leq 0.9$. 
\end{proof}


\subsection{Proof of~\Cref{lem:evasiveness implies learnability intro}}

We have shown that there exists an $H$ that evades all samplable distributions. Recall the definition of $\mcC_H=\{f_H : f\in \mcC\}$; it is the restriction of an arbitrary concept class $\mcC$ to $H$. We now show that $\mcC_H$ is learnable in samplable PAC. 



\LearnEvasiveH*

\begin{proof}
Let $\mcD$ be any size-$s$ distribution. By assumption, this distribution $(\eps,k)$-misses $H$. Let $m= O(k/\eps)$ be the size of the sample $\bS$, and let $A$ be the learner that, given a sample $\bS\sim \mcD^m$, outputs a hypothesis $h$ with the following properties (1) $h$ correctly labels all $x$ in $\bS$ and (2) $h(x)=0$ for all $x$ not in $\bS$. This learner $A$ can be efficiently implemented by simply appending every example in $\bS$ to the end of a DNF. There are $m$ examples each of length $n$, so the total runtime is $O(mn)$. 

Let $\error_{\mcD}(A)$ denote $A$'s expected error over $\mcD$. Since the target function is constant 0 outside $H$, this learner can only err on $x$ if $x\in H$ and $x\not\in S$. Therefore, its expected error is at most 
\begin{align*}
     \error_{\mcD}(A)\leq \sum_{x \in H} \mathcal{D}(x) \cdot \Prx_{\bS\sim \mathcal{D}^m}[x \notin \bS].
\end{align*}

Let $H^*\subseteq H$ be the $k$ heaviest elements in $H$. Formally, $H^*$ is the set maximizing $\mcD(H^*)$ subject to the constraints $|H^*|\leq k$ and $H^*\subseteq H$. Then, we can split the above sum in two by considering elements in $H^*$ and not in $H^*$, 
\begin{align*}
     \error_{\mcD}(A) &= \sum_{x \in H\setminus H^*} \mathcal{D}(x) \cdot \Prx_{\bS\sim \mathcal{D}^m}[x \notin \bS]+\sum_{x \in H^*} \mathcal{D}(x) \cdot \Prx_{\bS\sim \mathcal{D}^m}[x \notin \bS] \\
     &\leq \eps +\sum_{x \in H^*} \mathcal{D}(x) \cdot \Prx_{\bS\sim \mathcal{D}^m}[x \notin \bS] \tag{Definition of $(\eps,k)$-missing}\\
      &\leq \eps +\sum_{x \in H^*} \mathcal{D}(x) \cdot (1 - \mcD(x))^m\\
      &\leq \eps +\sum_{x \in H^*} \mathcal{D}(x) e^{-\mcD(x)m}.
\end{align*}
By taking the first derivative of each term in the sum with respect to $\mcD(x)$ and setting it equal to 0, we can see that the above sum is maximized if each $\mcD(x)=1/m$, at which point this simplifies to
\begin{align*}
         \error_{\mcD}(A)&\leq \eps + \sum_{x \in H^*}\frac{1}{e m}\\
         &= \eps + \frac{k}{e m}. \tag{Size of $T$}
\end{align*}
Setting $m= O(k/\eps)$ achieves expected error $O(\eps)$. Applying a Markov bound completes the proof. 
\end{proof}


\subsection{Proof of~\Cref{thm:statistical formal}}

The proof is straightforward given the lemmas in the previous sections. Recall that $\mcA$ denotes the class of all functions $f:\zo^n\rightarrow \zo$. By~\Cref{lem:non-explicit intro}, there exists an $H$ of density $\ge \delta$ that $(\eps, O((s\log s)/\eps))$-evades all size-$s$ distributions. Consider the concept class $\mcA_{H}$. 

\paragraph{Proof of (i)}  Clearly, $H$ is shattered by $\mcA_H$. Since $H$ has size $ \ge \delta 2^n$ (and no set larger than $H$ can be shattered by $\mcA_H$),  this means $\mcA_H$ has VC dimension $ \ge \delta 2^n$. 
\paragraph{Proof of (ii)} Because $H$ does $(\eps, O((s\log s)/\eps))$-evade all size-$s$ distributions, we can apply~\Cref{lem:evasiveness implies learnability intro}, to conclude that there is an algorithm that learns $\mcA_H$ to error $O(\eps)$ with sample complexity $O((s\log s)/\eps^2)$ and runtime $ O(n(s\log s)/\eps^2)$.

\section{Pseudorandom constructions of evasive sets and the proof of~\Cref{thm:separations within samplable PAC intro}}

In this section we use bounded-independence pseudorandom generators to prove~\Cref{thm:separations within samplable PAC intro}. Our techniques will also be useful for our proof of computational analogue of~\Cref{thm:separations within samplable PAC intro} in~\Cref{sec:computational-within-samplable}.


\begin{theorem}[Formal version of \Cref{thm:separations within samplable PAC intro}]
    \label{thm:separations within samplable PAC formal}
    For every $t\ge n$ and $\delta \le  2^{-(n/2+1)}$, there exists a concept class $\mcC$ such that the following holds.
    \begin{enumerate}
        \item There is a size-$O(tn\log t)$ distribution $\D$ such that any algorithm that learns $\mcC$ over $\D$ to error $0.1$ requires $\Omega(\delta 2^n)$ many samples. 
        \item For all $s\le t$ and $\eps \ge 4\delta$, there is an efficient algorithm that learns $\mcC$ to error $\eps$ over size-$s$ distributions using $O( (s\log s)/\eps^2)$ samples.
    \end{enumerate}
\end{theorem}

\Cref{thm:separations within samplable PAC intro} follows immediately from \Cref{thm:separations within samplable PAC formal} by choosing $\delta=2^{-(n/2+1)}$ and $t=s$. With this $\delta$ and $t$, we have that $\mcC$ is learnable over size-$s$ distributions to error $\eps$ using $O((s\log s)/\eps^2)$ samples yet requires $2^{\Omega(n)}$ samples to learn even just to error $0.1$ over distributions of size $S = \Omega(sn\log s)$.

\subsection{An explicit {\sl partially}-evasive set}

A key lemma is the construction of an explicit set $H$ that is {\sl partially} evasive in the sense that it evades all distributions of size smaller than the circuit that decides membership in $H$: 

\begin{lemma}
    \label{lem:PRG evasive formal}
    For every $\delta\geq 2^{-n}$, and $t\ge n$, there exists a set $H\sse\zo^n$ of density $\delta$ such that the following holds.
    \begin{enumerate}
        \item There is a circuit $C$ of size 
        \red{$O(tn\log^2 t)$} 
        such that $C(x)=\Ind[x\in H]$ for all $x\in\zo^n$.
        \item $H$ is a set that $(\eps, O((s\log s) / \eps ))$-evades all size-$s$ distributions for $s\le t$ and $\eps\ge 4\delta$.
    \end{enumerate}
\end{lemma}

\paragraph{Pseudorandom construction of $H$ for \Cref{lem:PRG evasive formal}.}
{
\Cref{fact:explicit prgs for rvs} provides a straightforward way of randomly constructing a subset $\bH\sse \zo^n$: set $N = 2^n$, sample $\br\sim \zo^{\ell}$, and interpret $\mcG(\br)\in \zo^N$ as the indicator string for a subset $H_{\br}\sse \zo^n$. The fact that $\mcG$ generates a $\delta$-biased $d$-wise independent distribution over $\zo^N$ implies that the random variables $\{\mathds{1}[x\in \bH]\}_{x\in \zo^n}$ are $\delta$-biased $d$-wise independent.

The explicitness of the PRG ensures that for every fixed $r\in \zo^\ell$, there is a circuit of size $O(d n\log(dn))$ that computes membership in $H_r$. We will show separately that (1) any fixed distribution hits $\bH$ with small probability (\Cref{claim:evasiveness for limited independence}) and (2) $\bH$ is $\delta/2$-dense with high probability (\Cref{cor:cantelli}). The final proof of the lemma will then follow by a union bound which shows there exists a fixed set $H$ which is both dense and evades all small-size distributions. 
}

\begin{claim}[Fixed distribution hits a pseudorandomly generated $\bH$ with small probability]
\label{claim:evasiveness for limited independence}
    Let $\delta<1/2$ and let $\bH\sse\zo^n$ be $\delta$-biased $d$-wise independent subset. Let $\D$ be a distribution over $\zo^n$. Then,
    $$
    \Prx_{\bH}[\D\text{ }(\eps,k)\textnormal{-hits } \bH]\le e^{-k\eps/6}
    $$
    for any $\eps\ge 2\delta$ and $k\le \min\{d/(2\delta), 2^n/2\}$.
\end{claim}
\begin{proof}
    We mostly follow the proof of \Cref{claim:fully-random misses} but will use concentration for random variables with bounded independence (\Cref{fact:SSS}) rather than a standard Chernoff bound. Let $H^*\sse \zo^n$ be the set of $x$ such that $\D(x)\ge 1/k$. Since $\D$ is a distribution there are at most $k$ such $x$: $|H^*|\le k$. If $\D$ $(\eps,k)$-hits $\bH$, then $\D(\bH\setminus H^*)\ge \eps$. 

    Let $\overline{H^*}\coloneqq \zo^n\setminus H^*$ denote the complement of $H^*$. For each $x\in\overline{H^*}$, we define a scaled indicator random variable $\bZ_x\coloneqq k\cdot \D(x)\cdot\Ind[x\in\bH]$ so that $\bZ_x\in [0,1]$ and
    $$
    \sum_{x\in\overline{H^*}} \bZ_x = k\cdot \D(\bH\setminus H^*).
    $$
    Furthermore, we compute the mean of the sum as 
    \begin{align*}
        \mu\coloneqq\E\Bigg[\sum_{x\in \overline{H^*}} \bZ_x \Bigg]&=k\cdot\sum_{x\in \overline{H^*}} \D(x)\Pr[x\in\bH]\tag{Definition of $\bZ_x$}\\
        &=d\delta \sum_{x\in \overline{H^*}} \D(x)\le k\delta.\tag{Assumption on $\bH$ and $\D$ is a distribution}
    \end{align*}
    Using our assumptions that $k\le 2^n/2$ and $\delta < 1/2$, it is straightforward to verify that
    $$
    \frac{\mu }{ 1- \mu/|\overline{H^*}| }\le 2k\delta
    $$
    and therefore we have enough independence to apply \Cref{fact:SSS}. Let $\gamma$ be such that $k\eps \coloneqq (1+\gamma)\mu$. Notice by our assumption that $\eps \ge 2\delta$, we have $k\eps \ge 2\delta k \ge 2\mu$ and so $\gamma \ge  1$. We calculate
    \begin{align*}
        \Pr[\D(\bH\setminus H^*)\ge \eps] &= \Pr[k\D(\bH\setminus H^*)\ge k\eps]\\
        &=\Pr\bracket[\Bigg]{\,\sum_{x\in\overline{H^*}} \bZ_x \ge k\eps }\\
        &\le \left(\frac{e^\gamma}{(1+\gamma)^{(1+\gamma)}}\right)^\mu\tag{\Cref{fact:SSS}}\\
        &= \exp\left(\left(\frac{\gamma}{1+\gamma} - \ln(1+\gamma)\right)k\eps\right)\tag{$\mu = k\eps/(1+\gamma)$}\\
        &<\exp(-k\eps/6)\tag{$\gamma\ge 1$}
    \end{align*}
    where in the last line we used the fact that the function $\gamma / (1+\gamma) - \ln(1+\gamma)$ for $\gamma\in [1,\infty)$ has its maximum at $\gamma=1$ and $1/2-\ln(2) < -1/6$. 
\end{proof}

\begin{claim}
    Let $\bZ$ have the same mean and variance as a binomial distribution with mean $\mu \geq 1$. Then $\Pr[\bZ > \mu/2] \geq 1/5$.
\end{claim}
\begin{proof}
    We apply Cantelli's inequality (sometimes referred to as one-sided Chebyshev's inequality): For any $t \geq 0$
    \begin{equation*}
        \Pr[\bZ \leq \mu - t ] \leq \frac{\Var[\bZ]}{\Var[\bZ] + t^2}.
    \end{equation*}
    Using the upper bound $\Var[\bZ] \leq \mu$ and plugging in $t = \mu/2$,
    \begin{equation*}
        \Pr[\bZ \leq \mu/2] \leq \frac{\mu}{\mu + (\mu/2)^2} = \frac{4}{4 + \mu} \leq \frac{4}{5}. \qedhere
    \end{equation*}
\end{proof}
Using the fact that the mean and variance of $\sum_{i \in [m]} \bz_i$ are the same as if $\bz_1, \ldots, \bz_m$ are fully independent or pairwise independent, we immediately obtain the following.
\begin{corollary}
\label{cor:cantelli}
    Let $\bH\sse\zo^n$ be a random subset such that the random variables $\{\mathds{1}[x\in \bH]\}_{x\in \zo^n}$ are $\delta$-biased, $d$-wise independent for $d\ge 2$ and $\delta \geq 2^{-n}$. Then, $|\bH|>\delta 2^{n-1}$ with probability at least $1/5$.
\end{corollary}

We can now prove the main lemma of this subsection.

\begin{proof}[Proof of \Cref{lem:PRG evasive formal}]

\red{Let $\delta' = 2\delta$ and} let $\bH\sse\zo^n$ be a random subset generated by the PRG $\mcG$ from \Cref{fact:explicit prgs for rvs} so that the random variables $\{\mathds{1}[x\in \bH]\}_{x\in \zo^n}$ are \red{$\delta'$}-biased $d$-wise independent for $d= c\cdot t\log t$ where $c$ is a large constant chosen later. We will show that with nonzero probability, $\bH$ satisfies all the requirements in the lemma statement. Specifically, we will show that the following probability is strictly less than $1$: 
\red{
$$
\Pr\Bigg[\exists s \le t, i\in [n]: 2^{-i} \ge 2\delta' \text{ and } \Big[|\bH|\le \delta' 2^{n-1}\text{ or some size-}s\text{ }\D\text{ }(2^{-i},2^i cs\log s) \text{-hits }\bH\Big]\Bigg].
$$
}
Since \red{$\delta'\geq 2^{-n}$}, we have that 
\red{$2^{-i} \ge 2\delta'  > 2^{-n}$}.
Thus, by~\Cref{prop:discrete_eps}, if $H$ is a set that $(2^{-i}, \tfrac{cs\log s}{2^{-i}})$-evades $\D$ for all $i\in [n]$ with \red{$2^{-i}\ge 2\delta'$}, then $H$ also $(\eps, O(\tfrac{cs\log s}{\eps}))$-evades $\D$ for all \red{$\eps\ge 2\delta'$}. We calculate
\red{\begin{align*}
        \Pr\Bigg[\exists s \le t, i\in [n]: &2^{-i} \ge 2\delta' \text{ and } \Big[|\bH|\le \delta' 2^{n-1}\text{ or some size-}s\text{ }\D\text{ }(2^{-i},2^i cs\log s) \text{-hits }\bH\Big]\Bigg]\\
        &\le\Pr\big[|\bH|\le \delta' 2^{n-1}\big] + \sum_{i=1}^{\log(1/2\delta)}\sum_{n\le s\le t}\Pr\Big[\text{some size-}s\text{ }\D\text{ }(2^{-i},\tfrac{cs\log s}{2^{-i}}) \text{-hits }\bH\Big]\tag{Union bound}\\
        &\le \frac{1}{5}+\sum_{i=1}^{\log(1/2\delta')}\sum_{n\le s\le t}\sum_{\text{size-}s\text{ }\D}\Pr\Big[\D\text{ }(2^{-i},\tfrac{cs\log s}{2^{-i}}) \text{-hits }\bH\Big]\tag{Union bound and \Cref{cor:cantelli}}\\
        &\le \frac{1}{5} + \sum_{i=1}^{\log(1/2\delta')}\sum_{n\le s\le t}\sum_{\text{size-}s\text{ }\D} e^{- (c s\log s) / 3}\tag{\Cref{claim:evasiveness for limited independence}}\\
        &\le \frac{1}{5} + \sum_{i=1}^{\log(1/2\delta')}\sum_{n\le s\le t}s^{O(s)}\cdot e^{- (c s\log s) / 3}\tag{\Cref{fact:circuits}}\\
        &\le \frac{1}{5}+ \sum_{i=1}^{\log(1/2\delta')}\sum_{n\le s\le t}s^{-\Omega(s)} = 1/5 + n^{-\Omega(n)} < 1.\tag{Large enough choice of $c$}
    \end{align*}}
    This shows that there exists some $H$ which simultaneously has \red{$|H|> \delta' 2^{n-1} = \delta 2^n$} and also $(\eps, O(s\log s )/\eps)$-evades all size-$s$ distributions for $s\le t$ and \red{$\eps\ge 2\delta' = 4\delta$}. Since this $H$ is generated by the PRG $\mcG$, there is a circuit of size 
    \red{$O(d n\log(dn)) = O(tn\log^2 t)$}
    that computes membership in $H$. 
\end{proof}


\subsection{A variant of~\Cref{lem:PRG evasive formal}}

In this section we prove a variant of~\Cref{lem:PRG evasive formal} that will be useful for our proof of the computational analogue of~\Cref{thm:separations within samplable PAC intro} in~\Cref{sec:computational-within-samplable}. 

\begin{lemma}
    \label{lem:samplability hierarchy formal}
    For every $\delta\le 2^{-(n/2+1)}$ and $t\ge n$ satisfying $t\log t\le o(\delta 2^n) $, there exists a $\delta$-dense set $H\sse\zo^n$ such that the following holds.
    \begin{enumerate}
        \item There is a circuit $C:\zo^{\log |H|}\to \zo^n$ of size $O(tn\log t)$ that generates the distribution $\mathrm{Unif}(H)$.
        \item  $H$ $(\eps, O((s\log s)/\eps) )$-evades all size-$s$ distributions for $s\le t$ and $\eps\ge 4\delta$.
    \end{enumerate}
\end{lemma}

\paragraph{Pseudorandom construction of $H$ for \Cref{lem:samplability hierarchy formal}}{
We prove \Cref{lem:samplability hierarchy formal} using a pseudorandom construction of $H$ that is different from the one we used to prove \Cref{lem:PRG evasive formal}. Specifically, let $m$ be a parameter controlling the size of $H$. From \Cref{fact:explicit prgs for rvs}, there is an explicit PRG $\mcG:\zo^\ell\to(\zo^n)^m$ and so we can construct a random subset $H_{\br}\sse \zo^n$ by sampling $\br\sim\zo^\ell$ and setting $H_{\br}=\{\mcG(\br)^{(1)},\ldots,\mcG(\br)^{(m)}\}$. Since $\mcG$ is a PRG for $d$-wise independent random variables, the members of $H_{\br}$ are similarly $d$-wise independent. Furthermore, since $\mcG$ is explicit, for every fixed $r\in\zo^\ell$, there is a circuit $C:\zo^{\log m}\to\zo^n$ of size $O(\ell\log \ell)$ such that $C(i)$ is the $i$th member of $H_r$. Therefore, by choosing $\bi\sim\zo^{\log m}$, we see that $C$ is a generator for $\mathrm{Unif}(H_r)$, assuming all members are unique. Following the approach used to prove \Cref{lem:PRG evasive formal}, we will separately prove that (1) the size of $H_{\br}$ is exactly $m$ with high probability (\Cref{claim:H is large}) and that (2) a distribution hits $H_{\br}$ with small probability (\Cref{claim:samplable H is evasive}).
}


\begin{claim}[A pseudorandomly generated $\bH$ is maximally large with high probability]
\label{claim:H is large}
    Let $\bx^{(1)},\ldots,\bx^{(m)}$ be random variables over $\zo^n$ which are marginally uniform and $d$-wise independent for $d\ge 2$ and $m\le 2^{n/2-1}$. Then, with probability at least $1/2$, the set $\bH=\{\bx^{(1)},\ldots,\bx^{(m)}\}$ has size exactly $m$.
\end{claim}

\begin{proof}
    To show that $\bH$ has size $m$, we will bound the number of $\bx^{(i)}$ which are duplicated. To start:
    \begin{equation}
    \label{eq:size of H}
        |\bH|\ge m -\sum_{i=1}^m\Ind[\bx^{(i)} = \bx^{(j)}\text{ for some }j\neq i].
    \end{equation}
    So it is sufficient to show that with probability at least $1/2$, the sum on the RHS is less than $1$ (and therefore $0$). 
    We have
    \begin{align*}
        \E\left[\sum_{i=1}^m\Ind[\bx^{(i)} = \bx^{(j)}\text{ for some }j\neq i]\right]&\le \sum_{i\neq j}\Pr[\bx^{(i)}=\bx^{(j)}]\tag{Union bound}\\
        &\le\frac{m^2}{2^n}\tag{$d$-wise independence for $d\ge 2$}\\
        &\le \frac{1}{2}\tag{Assumption that $m\le 2^{n/2-1}$}.
    \end{align*}
    Therefore, by Markov's inequality: 
    $$
    \Pr\left[ \sum_{i=1}^m\Ind[\bx^{(i)} = \bx^{(j)}\text{ for some }j\neq i] < 1\right]> \frac{1}{2}
    $$
    which completes the proof when combined with \Cref{eq:size of H}.
\end{proof}

\begin{claim}[A distribution hits a pseudorandomly generated $\bH$ with small probability]
    \label{claim:samplable H is evasive}
    Let $\bx^{(1)},\ldots,\bx^{(m)}$ be random variables over $\zo^n$ which are marginally uniform and $d$-wise independent. Let $\bH=\{\bx^{(1)},\ldots,\bx^{(m)}\}$. We have
    $$
    \Pr[\D\text{ }(\eps,k)\text{-hits }\bH]\le \exp(-k\eps/6)
    $$
    for any $\eps\ge 2m/2^n$ and $k\le \min\{d\cdot 2^n/(2m),2^n/2\}$.
\end{claim}

\begin{proof}
    Let $H^*\sse\zo^n$ be the set of points $x\in\zo^n$ such that $\D(x)\ge 1/k$. Since $\D$ is a distribution, we have $|H^*|\le k$. Therefore, if $\D$ is a distribution that $(\eps,k)$-hits $\bH$, we have $\D(\bH\setminus H^\star)\ge \eps$ and in particular:
    $$
    \sum_{i\in [m]} \D(\bx^{(i)})\Ind[\D(\bx^{(i)})\le 1/k]\ge \sum_{\bx\in \bH}\D(\bx)\Ind[\D(\bx)\ge 1/k]\ge \eps.
    $$
    Let $\bZ_i\coloneqq k\cdot \D(\bx^{(i)})\Ind[\D(\bx^{(i)})\le 1/k]$. The above shows that it is sufficient to bound
    $$
    \Pr\left[\sum_{i\in [m]} \bZ_i\ge k\eps\right]\le\exp(-k\eps). 
    $$
    Since $\bx^{(1)},\ldots,\bx^{(m)}$ are $d$-wise independent, the random variables $\{\bZ_i\}_{i\in [m]}$ are $d$-wise independent and also bounded in the range $[0,1]$. We compute the mean of the sum of the $\bZ_i$ as
    \begin{align*}
        \mu\coloneqq \E\left[\sum_{i\in [m]}\bZ_i\right]&=k\sum_{i\in [m]}\E\big[\D(\bx^{(i)})\Ind[\D(\bx^{(i)}) \le 1/k]\big]\\
        &=k\sum_{i\in [m]} \sum_{x\in \zo^n}\frac{1}{2^n}\cdot \D(x)\Ind[\D(x)\le 1/k] \tag{Each $\bx^{(i)}$ is marginally uniform}\\
        &\le \frac{km}{2^n}\tag{$\D$ is a distribution}.
    \end{align*}
    Let $\gamma$ be such that $k\eps = (1+\gamma)\mu$. By our assumption that $\eps\ge 2m/2^n$, we have $k\eps\ge 2\mu$ and so $\gamma \ge 1$. Also, using our assumptions that $d\ge 2km/2^n \ge 2\mu$ and $k\le 2^n/2$, it is straightforward to verify that 
    $$
    d\ge \frac{\gamma\mu }{1-\mu/m}
    $$
    and therefore we have enough independence to use \Cref{fact:SSS}. We calculate:
    \begin{align*}
        \Pr[\D\text{ }(\eps,k)\text{-hits }\bH]&\le \Pr\left[\sum_{i\in [m]}\bZ_i\ge k\eps \right]\\
        &\le \left(\frac{e^\gamma}{(1+\gamma)^{(1+\gamma)}}\right)^\mu\tag{\Cref{fact:SSS}}\\
        &< \exp(-k\eps/6)
    \end{align*}
    where the last line used the same inequality as derived in the proof of \Cref{claim:evasiveness for limited independence}.
\end{proof}

We are now able to prove \Cref{lem:samplability hierarchy formal}.

\begin{proof}[Proof of \Cref{lem:samplability hierarchy formal}]
    Let $m\coloneqq \delta 2^n$. Let $\mcG:\zo^\ell\to(\zo^n)^m$ be the explicit PRG from \Cref{fact:explicit prgs for rvs} for $d=c\cdot t\log t$ where $c$ is a large constant chosen later and let $\bH = \{\mcG(\br)^{(1)},\ldots,\mcG(\br)^{(m)}\}$ for a uniform random $\br\sim\zo^\ell$. Since $\mcG$ is a PRG for $d$-wise independent random variables, the elements of $\bH$ are similarly $d$-wise independent. As in the proof of \Cref{lem:PRG evasive formal}, it is sufficient for us to show that 
    $$
    \Pr\Big[|\bH|< m\text{ or some size-}s\text{ }\D\text{ }(\eps,\tfrac{cs\log s}{\eps} ) \text{-hits }\bH\text{ for }s\le t\text{ and }\eps=2^{-i} \ge 2\delta\text{ for }i\in [n]\Big]<1.
    $$
    We have
    \begin{align*}
        \Pr\Big[|\bH|&<m \text{ or some size-}s\text{ }\D\text{ }(\eps,\tfrac{cs\log s}{\eps} ) \text{-hits }\bH\text{ for }s\le t\text{ and }\eps=2^{-i} \ge 2\delta\text{ for }i\in [n]\Big]\\
        &\le \Pr\bracket[\Big]{|\bH|<m]} + \sum_{i=1}^{\log(1/2\delta)}\sum_{n\le s\le t}\Pr\Big[\text{some size-}s\text{ }\D\text{ }(2^{-i},\tfrac{cs\log s}{2^{-i}}) \text{-hits }\bH \Big] \tag{Union bound}\\
        &\le \frac{1}{2} + \sum_{i=1}^{\log(1/2\delta)}\sum_{n\le s\le t}\sum_{\text{size-}s \text{ }\D}\Pr\bracket[\Big]{\D\text{ }(2^{-i},\tfrac{cs\log s}{2^{-i}})\text{-hits }H} \tag{Union bound and \Cref{claim:H is large} since $\delta\le 2^{-(n/2+1)}$}\\
        &\le \frac{1}{2}+\sum_{i=1}^{\log(1/2\delta)}\sum_{n\le s\le t}\sum_{\text{size-}s \text{ }\D}e^{-(cs\log s)/3}\tag{\Cref{claim:samplable H is evasive}}\\
        &\le \frac{1}{2}+\sum_{i=1}^{\log(1/2\delta)}\sum_{n\le s\le t} s^{O(s)}\cdot e^{-(cs\log s)/3}\tag{\Cref{fact:circuits}}\\
        &\le \frac{1}{2} +\sum_{i=1}^{\log(1/2\delta)} \sum_{n\le s\le t}s^{-\Omega(s)}=\frac{1}{2}+n^{-\Omega(n)}\tag{Large enough choice of c}.
    \end{align*}
    We used the fact that we can apply \Cref{claim:samplable H is evasive} to bound $\Pr[\D\text{ }(2^{-i},\tfrac{cs\log s}{2^{-i}})\text{-hits }H]$. This is because of our choice of $d=ct\log t$, $m=\delta 2^n$ and the assumption that $t\log t=o(m)$, from which it is straightforward to verify that we fulfill the parameter requirements of \Cref{claim:samplable H is evasive}. The above calculation shows that there is a fixed $H$ which both has size $|H|= m$ and $(\eps,O(s\log (s)/\eps))$-evades all size-$s$ distributions for $s\le t$ and $\eps\ge 2\delta$. Since $H$ is generated by the PRG $\mcG$, there is a circuit of size $O(dn\log (dn)) = O(tn\log t)$ which generates the distribution $\mathrm{Unif}(H)$.
\end{proof}

\subsection{Putting everything together: Proof of \Cref{thm:separations within samplable PAC formal}}
Let $H\sse\zo^n$ be the set of inputs from \Cref{lem:samplability hierarchy formal}. We prove the theorem for the concept class $\mcA_H$.  Recall that $\mcA$ denotes the class of all functions $f:\zo^n\to\zo$ and $\mcA_H$ is the restriction of this class to the set $H$. Since $H$ is a set that $(\eps, O(s\log(s)/\eps))$-evades all size-$s$ distributions for any $s\le t$ and $\eps\ge 2\delta$, we get by \Cref{lem:evasiveness implies learnability intro} that there is an efficient algorithm that learns $\mcA_H$ over any size-$s$ distribution to error $\eps$ using $O(s\log (s)/\eps^2)$ samples. This proves part (ii) of the theorem statement. It remains to prove (i).

By definition, $H$ is a shattering set of size $\delta 2^n$. Let $\D$ be the distribution over $H$ from \Cref{fact:shattering} for $\eps\coloneqq 0.1$. Since any algorithm that learns $\mcA_H$ over $\D$ requires $\Omega(\delta 2^n / \eps)$, we only need to show that $\D$ can be generated by a circuit of size $O(tn\log t+\log (1/\eps))$. From \Cref{lem:samplability hierarchy formal}, let $C:\zo^{\log |H|}\to\zo^n$ be the circuit of size $O(tn\log t)$ that generates the uniform distribution over $H$.

Let $r\in \N$ be such that $8\eps=\tfrac{1}{2^r} - \frac{1}{|H|}$.  It is straightforward to verify that $r\le O(\log (\tfrac{1}{|H|\eps}))= O( n + \log (1/\eps))$. We define a circuit $C':\zo^{\log |H|+r}\to\zo^n$ as follows. Let $x^\star\in H$ denote the heaviest element under $\D$ and let $i^\star\in \zo^{\log |H|}$ denote the unique input such that $C(i^\star)=x^\star$. On input $(i,z)\in \zo^{\log |H|}\times \zo^{r}$, the circuit $C'$ computes the following
$$
C'(i, z)\coloneqq\begin{cases}
    C(i^\star) & \text{if }i=i^\star \text{ or }z\neq 0^r\\
    C(i) & \text{otherwise}
\end{cases}.
$$
It is straightforward to verify that $C'$ can be written as a circuit of size $O(|C|+\log (1/\eps))=O(tn\log t+\log (1/\eps))$. We claim that $C'$ generates the distribution $\D$. We have
\begin{align*}
    \Prx_{(\bi,\bz)\sim \zo^{\log |H|}\times \zo^{r}}[C'(\bi,\bz) = x^\star] &= \Pr[\bi=i^\star\text{ or }\bz\neq 0^r] \\
    &= \frac{1}{|H|} + 1-\frac{1}{2^r} \\
    &= 1-8\eps
\end{align*}
by our choice of $r$. When $C'$ doesn't output $x^\star$, its output is uniform over $H\setminus \{x^\star\}$. By the above calculation, this happens with probability $8\eps$. Therefore, $C'$ generates $\D$ as desired.\hfill\qed

\section{\Cref{key assumption intro} and the proofs of computational separations}
\label{sec:conjecture}
We introduce the hardness assumption that there exists explicit evasive sets: 
\begin{conjecture}[Formal version of \Cref{key assumption intro}]
    \label{conjecture:hard-to-sample}
    There exists sets $\set{H_n \subseteq \zo^n}_{n \in \N}$ \red{and constant $c$} satisfying the following.
    \begin{enumerate}
        \item Explicit: Membership in $H_n$ is computed by a circuit of size $\poly(n)$.
        \item Large: $H_n$ contains superpolynomially many points.
        \item Evasive: 
        For all $s = \poly(n)$, $H_n$ $(1/s, s^c)$-evades all size-$s$ distributions.
    \end{enumerate}    
\end{conjecture}
As far as we know, \Cref{conjecture:hard-to-sample} could hold with stronger parameters. For example, if the best strategy for sampling $H$ is to memorize $\approx s/n$ elements of $H$, then \Cref{conjecture:hard-to-sample} holds even with $c = 1$.

We use~\Cref{conjecture:hard-to-sample} in the formal version of \Cref{thm:computational intro}.
\begin{theorem}[Formal version of \Cref{thm:computational intro}]
    \label{thm:computational-body}
    If \Cref{conjecture:hard-to-sample} is true and one-way functions exist, there is a concept class $\mcC$ with the following properties.
    \begin{enumerate}
        \item $\mcC$ is a subclass of polynomial-size circuits.
        \item $\mcC$ is easy to learn in samplable PAC: For every $s(n) \leq \poly(n)$ and $\eps(n) \geq 1/\poly(n)$, there is a polytime algorithm that learns $\mcC$ to accuracy $(1-\eps(n))$ on size-$s(n)$ distributions.
        \item $\mcC$ is hard to learn in standard PAC: There is a distribution under which learning $\mcC$ to constant accuracy requires superpolynomial time.
    \end{enumerate}
\end{theorem}
\begin{remark}[The quantitative time lower bound in standard PAC]
    \label{remark:poly-vs-exp}
    \Cref{thm:computational-body} obtains ``only" a superpolynomial time lower bound in standard PAC. This stems from only assuming one-way functions that defeat all polynomial time adversaries. With stronger assumptions, this time lower bound could be easily strengthened to $2^{\Omega(n)}$. For that version, we would need to assume one-way functions for which, all time $O(2^{c n})$ adversaries can only achieve advantage $2^{-cn}$ for some constant $c$. We would also need to slightly strengthen \Cref{conjecture:hard-to-sample} to hold for larger sets, $H_n$ containing $2^{cn}$ points. In this section, we just prove the superpolynomial time lower bound using the more standard cryptographic assumption as the exponential time lower bound has essentially the same proof.
\end{remark}

\noindent The remainder of this section is structured as follows:
\begin{enumerate}
    \item In \Cref{subsec:unconditional-difficult}, we show that it is likely difficult to unconditionally prove \Cref{conjecture:hard-to-sample} by showing that doing so requires separating $\P$ from $\NP$.
    \item In \Cref{subsec:hard-to-sample-relative-to-random-oracle}, we provide formal evidence in favor of~\Cref{conjecture:hard-to-sample}, showing that it is true relative to a random oracle.
    \item In \Cref{subsec:proof of computational seperation}, we prove \Cref{thm:computational-body}.
    \item In \Cref{sec:computational-within-samplable}, we sketch the proof of \Cref{thm:computational analogue of separations within samplable PAC formal}, the computational analogue of \Cref{thm:separations within samplable PAC intro}. Its proof is essentially the same as \Cref{thm:computational-body}. 
\end{enumerate}

\subsection{\Cref{conjecture:hard-to-sample} implies $\P \neq \NP$}
\label{subsec:unconditional-difficult}


\begin{claim}
    \label{claim:hard-to-prove}
    If \Cref{conjecture:hard-to-sample} is true then $\P \neq \NP$.
\end{claim}
To prove \Cref{claim:hard-to-prove}, we use the following classic result of Jerrum, Valiant, and Vazirani.
\begin{fact}[Approximate sampling with an $\NP$ oracle~\cite{JVV86}]
    \label{fact:approx-sampling-NP-oracle}
    There exists a polynomial time (randomized) Turing machine equipped with an $\NP$ oracle that, given as input a circuit $\phi$, outputs a random variable $\bx$ satisfying, for any input $x$ that $\phi$ accepts,
    \begin{equation*}
        \frac{0.9}{|\phi^{-1}(1)|} \leq \Prx[\bx = x] \leq \frac{1.1}{|\phi^{-1}(1)|}.
    \end{equation*}
\end{fact}
We now prove the main result of this subsection.
\begin{proof}[Proof of \Cref{claim:hard-to-prove}]
    We will prove that if $\P = \NP$, then \Cref{conjecture:hard-to-sample} is false. The desired result follows by contrapositive.

    If $\P = \NP$, then in \Cref{fact:approx-sampling-NP-oracle}, the Turing machine need not have an $\NP$ oracle (as it can simulate this oracle itself). Hence, there is a randomized Turing machine $T$ that on input $\phi$ outputs an approximately uniform accepting input to $\phi$.

    Now, consider any sequence of sets $\set{H_n \subseteq \zo^n}_{n \in \N}$. We will show this sequence does not satisfy the requirements of \Cref{conjecture:hard-to-sample} by showing that if the sequence is explicit and large it is not evasive.

    Since $H_n$ is assumed to be explicit, it is the set of accepting inputs of some polynomial-sized circuit $\phi$. Then, $T(\phi)$ samples an approximately uniform element of $H_n$ in the sense that, for all $x \in H_n$,
    \begin{equation*}
        \Pr[T(\phi) = x] \geq \frac{0.9}{|H_n|}.
    \end{equation*}
    Using the Cook--Levin reduction, we may transform $T(\phi)$ into a circuit with input equal to the random bits that $T$ uses (and $\phi$ hard coded into the circuit). Hence, there is a poly-sized generator $\mcG$, that given as input a uniformly random seed $\br$ satisfies for all $x \in H_n$,
    \begin{equation*}
        \Pr[\mcG(\br) = x] \geq \frac{0.9}{|H_n|}.
    \end{equation*}
    We claim that the distribution $\mcG$ generates $(0.45, |H_n|/2)$-hits $H_n$. This is because, for any set $H^\star$ of size at most $|H_n|/2$, we bound
    \begin{equation*}
        \Pr[\mcG(\br) \in H_n \setminus H^{\star}] \geq \frac{0.9}{|H_n|} \cdot |H_n \setminus H^{\star}| \geq 0.45.
    \end{equation*}
    Let $s = \poly(n)$ be the size of the generator $\mcG$. Then, since $|H_n|$ grows superpolynomially and $s$ only grows polynomially, for any constant $c$ and large enough $n$, we have the distribution $\mcG$ generates $(0.45, s^c)$-hits $H_n$, which also implies it $(1/s, s^c)$-hits $H_n$. This contradicts the evasiveness requirement of \Cref{conjecture:hard-to-sample}.
    
    Therefore, if $\P = \NP$, \Cref{conjecture:hard-to-sample} is false.
\end{proof}

\subsection{\Cref{conjecture:hard-to-sample} is true relative to a random oracle}
\label{subsec:hard-to-sample-relative-to-random-oracle}
In this subsection, we prove the following.
\begin{theorem}[\Cref{conjecture:hard-to-sample} is true relative to a random oracle]
    \label{thm:hard-to-sample-relative-to-random-oracle}
    For large enough $n$, with probability at least $1 - o(1)$ over a uniform random oracle $\bmcO$, there is a set $H_{\bmcO}$ satisfying:
    \begin{enumerate}
        \item Explicit: There is an oracle circuit $C^{\bmcO}$ of size $O(n^2)$ that computes membership in $H_{\bmcO}$.
        \item Large: $H_{\bmcO}$ contains $2^{\Omega(n)}$ many unique inputs.
        \item Evasive: For every $n \leq s \leq 2^{O(n)}$, there is no size-$s$ oracle generator $G^{\bmcO}$ that $(1/s, s^{1.01})$-hits $H_{\bmcO}$.
    \end{enumerate}
\end{theorem}

\subsubsection{A helper lemma: Hardness of finding many satisfying assignments}
\begin{lemma}[Hardness of finding many distinct satisfying assignments]
\label{lem:oracle-LB}
For every $n$ and every $2^{-n} \leq \delta \leq 1/16$, there exists a size-$O(n \log(1/\delta))
)$ oracle circuit $C:\zo^n \to \zo$ such that, the following hold over the randomness of a uniform oracle $\bmcO:\zo^{n} \times \zo^{O(\log\log(1/\delta))} \to \zo$:
\begin{enumerate}
    \item For any (possibly randomized) $q$-query algorithm $\bT$ and $q \leq \frac{k}{16\delta}$, the probability that $\bT(\bmcO)$ outputs $k$ distinct accepting inputs of $C^{\bmcO}$ is at most $\exp(-\Omega(k))$.
	\item With probability at least $1 - \exp(-\Omega(\delta 2^n))$, $C^{\bmcO}$ accepts at least $\delta \cdot  2^n$ distinct inputs.
\end{enumerate}
\end{lemma}
When we apply \Cref{lem:oracle-LB}, we will union bound over $\exp(O(k))$ many size-$S$ circuits, which will allow us to ``swap the quantifiers;" i.e.~show that with high probability, $\bmcO$ is hard for every circuit of interest. We will do so in a setting where $S \gg k$, which means we could not simply union bound over all size-$S$ circuits.

\begin{claim}
\label{clm:oracle-binomial}
    Fix a randomized $q$-query algorithm $\bT$. If $\bmcO$ is a $\delta$-biased random oracle, then: 
    $$
    \Prx_{\bmcO}[\bT(\bmcO) \text{ outputs distinct } x^{(1)},\ldots,x^{(k)}\in \bmcO^{-1}(1)  ]\le \Pr[\mathrm{Bin}(q + k, \delta) \ge k]. 
    $$
\end{claim}
\begin{proof}

We will first prove a bound assuming that $\bT$ makes $q+k$ queries to $\bmcO$, that all of these $q+k$ queries are distinct elements of $\zo^n$, and also that all of $\bT(\bmcO)$ outputs are also among the $q +k$ queries that $\bT$ makes. In this case, the number of queries $\bT$ makes on which $\bmcO$ evaluates to $1$ is exactly distributed according to $\mathrm{Bin}(q+k, \delta)$. Since we assumed that $\bT$ only outputs elements it queried, to succeed, it must query at least $k$ distinct elements of $\bmcO^{-1}(1)$ which occurs with probability exactly $\Pr[\mathrm{Bin}(q + k, \delta) \ge k].$

We reduce any $q$-query algorithm $\bT$ to a $
\leq q+k$-query algorithm $\bT'$ with the desired structure. $\bT'$ first runs the $q$-queries of $\bT$, then queries the $k$ outputs of $\bT$, and finally returns the same outputs. This guarantees that $\bT'$ queries all of its outputs. Finally, we observe that if $\bT'$ queries the same input more than once, later queries are uninformative and can be removed. This gives a $\bT'$ with at most $q+k$ queries. Note that $\Pr[\mathrm{Bin}(m, \delta) \ge k]$ is increasing in $m$, so using $m = q+k$ as a conservative upper bound suffices.\end{proof}




\begin{proof}[Proof of \Cref{lem:oracle-LB}]
Let $\ell \coloneqq \floor{\log_2(1/\delta)} - 1$ and $\delta' = 2^{-\ell}$. This guarantees that $\delta' \in [2\delta, 4\delta]$. We will construct $C$ so that it accepts every input independently with probability exactly $\delta'$. Since $\bmcO$ accepts every element independently with probability $1/2$, we can accomplish this by having $C^{\bmcO}$ accept some $x$ iff $\ell$ distinct inputs are accepted by $\bmcO$.

Formally, let $r = \ceil{\log_2(\ell)}$. We have access to a uniformly random function $\bmcO:\zo^n \times \zo^r \to \zo$. Let $S \subseteq \zo^r$ be any fixed set of size $\ell$ (which is possible since $r \geq \log_2(\ell)$). Then we construct $C$ as
\begin{equation*}
    C^{\bmcO}(x) \coloneqq \bigwedge_{y \in S} \bmcO(x,y).
\end{equation*}
We briefly analyze the size of $C$: It consists of the input $x$ (which has size $n$) being fed into $\ell$ many instances of the subcircuit calling $\bmcO(x,y)$ for all $y \in S$. Each such call requires $n$ wires (for routing $x$) as well as $r$ many hardcoded values (for $y$). The \textsf{AND} gate at the top is a single gate, so the size of $C$ is $O(\ell(n+r))$. Since $\delta \geq 2^{-n}$ we have that $r \leq n$, so the size of $C$ is $O(\ell n) = O(n \log(1/\delta))$.

We also observe a key property about $C^{\bmcO}$. It independently accepts each $x$ with probability exactly $2^{-\ell} = \delta'$: This follows from $\bmcO$ being a uniformly random oracle and that $C$ queries $\ell$ distinct inputs to determine whether to accept each $x$. Using that observation, we proceed to show that $C^{\bmcO}$ has the desired properties:

\textbf{Item (1).} Consider any (possibly randomized) $q$-query algorithm $\bT$. We will strengthen $\bT$ in the following way: For any $x$, it may query $\bmcO(x,y)$ for all possible $y$ at a net cost of $1$ query. Since this only provides more power to $\bT$, our result will carry over to algorithms that can only query individual inputs to $\bmcO$ at one cost.

With this strengthening, we observe that one query of $\bT$ simulates one query to $C^{\bmcO}$, and gives no information about $C^{\mcO}$'s behavior on any other $x$. Furthermore, $C^{\bmcO}$ is exactly a $\delta'$-biased random oracle. Therefore, by \Cref{clm:oracle-binomial}
\begin{equation*}
    \Pr[\bT(\bmcO) \text{ outputs $k$ distinct accepting inputs to }C^{\bmcO}] \leq \Pr[\mathrm{Bin}(q + k, \delta') \geq k].
\end{equation*}
Since $\delta' \leq 4\delta \leq 1/4$ and $q \leq \frac{k}{16\delta}$, the expectation of $\mathrm{Bin}(q + k, \delta')$ is at most $k/4 + k/4 = k/2$. Therefore, by a standard Chernoff bound (\Cref{fact:chernoff}),
\begin{equation*}
     \Pr[\bT(\bmcO) \text{ outputs $k$ distinct accepting inputs to }C^{\bmcO}] \leq e^{-k/6}.
\end{equation*}



\textbf{Item (2)}.
As discussed earlier, each element of $\zo^n$ is included in $(C^{\bmcO})^{-1}(1)$ independently with probability $\delta' \geq 2\delta.$ Therefore, we wish to show that
\begin{equation*}
    \Pr[\mathrm{Bin}(2^n, \delta') < \delta 2^n] \leq \exp(-\delta \cdot 2^n / 4). 
\end{equation*}
which also follows by \Cref{fact:chernoff} and that $\delta' \geq 2\delta$.
\end{proof}

\subsubsection{Proof of~\Cref{thm:hard-to-sample-relative-to-random-oracle}: \Cref{conjecture:hard-to-sample} relative to a random oracle}



The key transition from \Cref{thm:hard-to-sample-relative-to-random-oracle} to \Cref{lem:oracle-LB} is the following simple result:
\begin{proposition}
   \label{prop:hitting-to-many}
    Let $\mcD$ be a distribution that $(\eps, k)$-hits a set $H$. Then, for \red{$m \coloneqq \ceil{2k/\eps}$}, if we draw $\by^{(1)}, \ldots, \by^{(m)} \iid \mcD$, then
    \begin{equation*}
        \Prx\bracket[\bigg]{\,\sum_{x \in H} \Ind[x \in \set{\by^{(1)}, \ldots, \by^{(m)}}] \geq k} \geq \frac{1}{2}.
    \end{equation*}
\end{proposition}
\begin{proof}
    For the sake of the argument, suppose rather than taking exactly $m$ samples $\mcD$, we continued taking i.i.d. samples until we have exactly $k$ unique elements of $H$. Let $\bz$ refer to the number of samples we take before stopping. We will argue that $\Ex[\bz] \leq \frac{k}{\eps}$. By Markov's inequality, this gives that $\bz \leq \frac{2k}{\eps}$ with probability at least $1/2$.

    To bound the expectation of $\bz$, we can write $\bz = \sum_{i \in [k]} \bz_i$ where $\bz_i$ is the number of additional samples it takes to find the $i^{\text{th}}$ unique element of $\bH$ after finding $i-1$ unique elements. We will show that $\Ex[\bz_i] \leq 1/\eps$ for all $i \in [k]$, giving the desired bound.

    For this, observe that whenever we take a sample $\by^{(j)}$, we are guaranteed that $\bH^{\star} \coloneqq \set{\by^{(1)}, \ldots, \by^{(j-1)}} \cap H$ contains at most $k$ elements. Therefore, by the definition of $(\eps, k)$-hitting, we have that $\mcD(H \setminus \bH^{\star}) \geq \eps$ and so each new sample has at least an $\eps$ probability of being a new unique element of $H$. This means that $\bz_i$ is upper bounded (stochastically dominated) by a geometric distribution with parameter $\eps$, and so its mean is upper bounded by $1/\eps$.
\end{proof}

\begin{proof}[Proof of \Cref{thm:hard-to-sample-relative-to-random-oracle}]
    We set $\delta \coloneqq 2^{-0.5n}$. Then, for any random oracle $\bmcO$, we set $H_{\bmcO}$ to be the accepting inputs of $C^{\bmcO}$ where $C$ is as defined in \Cref{lem:oracle-LB}. This immediately guarantees explicitness since $C$ has size $O(n^2)$. For largeness, \Cref{lem:oracle-LB} guarantees that with probability at least $1 - \exp(\Omega(2^{-0.5n}))$,  $C^{\bmcO}$ accepts at least $2^{0.5n}$ distinct inputs.

    All that remains is to show that $H_{\bmcO}$ is evasive with high probability. For this, we will show that for any fixed size-$s$ oracle generator $G$ where $s \leq 2^{0.24n}$, the probability that $G^{\bmcO}$ $(\eps \coloneqq 1/s, k \coloneqq s^{1.01})$-hits $H_{\bmcO}$ is at most $\exp(-\Omega(s^{1.01})$. For this, we first apply \Cref{prop:hitting-to-many}:
    \red{
    For any $\mcO$ for which $G^{\mcO}$ hits $H_{\mcO}$, if we take $m \coloneqq \ceil{2k/\eps}$ unique samples from $G^{\mcO}$, with probability at least $1/2$, they will contain at least $k$ unique elements of $H_{\mcO}$.
    Thus,
    $$\Prx_{\bmcO,\text{ sampling }}[m\text{ samples from }G^{\bmcO}\text{ contain } \ge k\text{ unique elements of }H_{\bmcO}] \ge \lfrac 12 \cdot \Prx_{\bmcO}[G^{\bmcO}\text{ hits }H_{\bmcO}].$$
    }
    
    Next, define $T_{G}(\bmcO)$ to be the $q \coloneqq O(m \cdot (s + n \log n))$ randomized query algorithm which does the following:
    \begin{enumerate}
        \item Draw $m$ independent random seeds, $\br_1, \ldots, \br_m$, for the generator $G$.
        \item Generate the samples $\bx_i = G^{\bmcO}(\br_i)$ for each $i \in [m]$. Each such sample requires at most $s$ queries to $\bmcO$.
        \item For each $\bx_i$, evaluate $C^{\bmcO}(\bx_i)$ using at most $O(n \log n)$ queries.
        \item If there are at least $k$ distinct accepting inputs of $C^{\bmcO}$ among these $m$ samples, output them. Otherwise, output arbitrarily.
    \end{enumerate}
     We note that, as required by \Cref{lem:oracle-LB}, our choice of parameters satisfies $k \geq 8\delta q$ because
     \begin{equation*}
         8\delta q = O(2^{-0.5n} \cdot m \cdot (s + n \log n)) = O(2^{-0.5n} \cdot ks \cdot (s + n\log n)),
     \end{equation*}
     which is at most $k$ for our choice of $s \leq 2^{0.24n}$ and large enough $n$. Hence, for any fixed choice of $G$, the probability that $T_{G}(\bmcO)$ successfully finds $k=s^{1.01}$ distinct accepting inputs to $C^{\bmcO}$ is at most $\exp(-\Omega(s^{1.01}))$. Union bounding over all $\exp(O(s \log s))$ size-$s$ circuits $G$, the probability there is a size-$s$ oracle generator $G^{\bmcO}$ that $(1/s, s^{1.01})$ hits $H_{\bmcO}$ is at most $\exp(O(s \log s)) \cdot \exp(-\Omega(s^{1.01}))$. In particular, for large enough $s$ (which is implied by our requirement that $n$ is large enough) this probability is at most $2^{-s}$. This failure probability is small enough to union bound over all choices of $n \leq s \leq 2^{0.24n}$.
\end{proof}


\subsection{Proof of \Cref{thm:computational-body}}
\label{subsec:proof of computational seperation}

\textbf{The concept class:} For each $n \in \N$, let $\mcF_n$ be a PRFF secure against non-uniform adversaries (\Cref{def:prff}) and $\set{H_n \subseteq \zo^n}_{n \in \N}$ be a sequence meeting the requirements of \Cref{conjecture:hard-to-sample}. We set the concept class to be $(\mcF_n)_{H_n}$ (recall \Cref{eq:concept-class-restriction} for this notation). The fact that $\mcC_n$ is a subclass of polynomial-size circuits follows from membership in $H_n$ being computable by polynomial-size circuits, every function in the PRFF being computable by a polynomial-size circuit, and that the logical AND of two polynomial-size circuits is itself a polynomial-size circuit.

\pparagraph{Easiness in samplable PAC:} Consider any $s(n) \leq \poly(n)$ and $\eps(n) \geq 1/\poly(n)$. Let $m(n) \coloneqq O(\max\{s(n), 1/\eps(n)\})$. Then, for some constant $c$ and large enough $n$, $H_n$ $(O(\eps(n)), m(n)^c)$-evades all size-$s(n)$-distributions. Hence, by \Cref{lem:evasiveness implies learnability intro}, we can efficiently learn $\mcC_n$ to accuracy $\eps(n)$ on all size-$s(n)$ distributions in time $O(n\cdot m(n)/\eps(n)) = \poly(n)$.

\pparagraph{Hardness in standard PAC:} As Valiant observed when defining the PAC model~\cite{Val84}, the very definition of pseudorandom functions families (PRFFs) (\Cref{def:prff}) implies that they are hard to learn under $\mathrm{Unif}(\zo^n)$. The hardness of learning $\mcC_n$ follows directly from a generalization of Valiant's observation, that PRFFs are in fact hard to learn under $\mathrm{Unif}(H)$ for any set $H\sse \zo^n$ with superpolynomially many points (in which case the quantity $\alpha$ below is small).
\begin{lemma}[Breaking PRFF security using efficient learners]
    \label{lem:PRFF-to-hardnes-of-learning}
    For any function family $\mcF \coloneqq \set{f_s: \zo^{n} \to \zo}_{s \in \zo^n}$ and distribution $\mcD$, if there is an efficient learner for $\mcF$ using $m$ samples on the distribution $\mcD$ that has expected accuracy $1 - \eps$, then there is a polynomial-time nonuniform adversary $A$ for which
    \begin{equation*}
        \begin{split}
        \left|
            \begin{aligned}
                    &\quad\Prx_{\bs\sim \zo^n}[A\text{ outputs }1\text{ when given oracle access to }f_{\bs} ]\\
                    &\quad \qquad -\Prx_{\boldf:\zo^n\to\zo}[A\text{ outputs }1\text{ when given oracle access to }\boldf ]\quad
                \end{aligned}
        \right| \geq \frac{1}{2} - \eps - \frac{\alpha}{2},
        \end{split}
    \end{equation*}
    where $\alpha \coloneqq \Prx_{\bx^{(1)}, \ldots, \bx^{(m+1)} \iid \mcD}\bracket*{\bx^{(m+1)} \in \set{\bx^{(1)}, \ldots, \bx^{(m)}}}$ and $\boldf:\zo^n\to\zo$ is a fully random  function.
\end{lemma}
\begin{proof}
    Let $L$ be the promised learner. By the expected accuracy of $L$, if we were to draw a seed $\bs \sim \Unif(\zo^n)$ and, for each $i \in [m+1]$, draw $\bx^{(i)} \iid \mcD$ and set $\by^{(i)} = f_{\bs}(\bx^{(i)})$, then
    \begin{equation*}
        \Prx_{\bh \leftarrow L((\bx^{(1)}, \by^{(1)}), \ldots, (\bx^{(m)}, \by^{(m)}))}[\bh(\bx_{m+1}) = \by_{m+1}] \geq 1-\eps.
    \end{equation*}
    On the other hand, if instead we set the label to $\bz^{(i)} \coloneqq \boldf(\bx^{(i)})$ for a fully random function $\boldf$, then as long as $\bx_{m+1}$ is not contained in $(\bx^{(1)}, \ldots, \bx^{(m)})$, then $\by_{m+1}$ will be independent of the sample $((\bx^{(1)}, \bz^{(1)}, \ldots, (\bx^{(m)}, \bz^{(m)}))$. This case occurs with probability at least $1 - \alpha$, and if it occurs, the learner has expected accuracy $1/2$. Therefore,
    \begin{equation*}
        \Prx_{\bh \leftarrow L((\bx^{(1)}, \bz^{(1)}), \ldots, (\bx^{(m)}, \bz^{(m)}))}[\bh(\bx_{m+1}) = \bz_{m+1}] \leq \alpha + \frac{1}{2} \cdot \paren*{1 - \alpha}= \frac{1}{2} + \frac{\alpha}{2}.
    \end{equation*}
    Let us say the \emph{advantage} of $(x^{(1)}, \ldots, x^{(m+1)})$ to be the expected difference of the above two quantities in expectation over the random seed $\bs$, fully random function $\boldf$, and any internal randomness of the learner. Then, comparing the two inequalities, we have that,
    \begin{equation*}
        \Ex_{\bx^{(1)}, \ldots, \bx^{(m+1)} \iid \Unif(H)}\bracket*{\text{advantage of }(\bx^{(1)}, \ldots, \bx^{(m+1)})} \geq 1-\eps -\paren*{\frac{1}{2} + \frac{\alpha}{2}} = \frac{1}{2} - \eps - \frac{ \alpha}{2}.
    \end{equation*}
    Hence, there exists an explicit choice of $(x^{(1)}, \ldots, x^{(m+1)})$ with advantage at least $\frac{1}{2} - \eps - \frac{\alpha}{2}$. Our nonuniform adversary will use its advice to hard code these choices. It runs the test suggested by the above analysis: It queries its oracle to attain labels $y^{(1)}, \ldots, y^{(m+1)}$, runs $L$ on the sample $(x^{(1)}, y^{(1)}), \ldots, (x^{(m)}, y^{(m)})$ to attain a hypothesis $h$, and then returns $\Ind[h(x^{(m+1)}) = y^{(m+1)}]$. By the prior analysis, this adversary had advantage at least $\frac{1}{2} - \eps - \frac{ \alpha}{2}$.
\end{proof}
To prove hardness in standard PAC, we observe that by \Cref{lem:PRFF-to-hardnes-of-learning}, if $\mcC$ were learnable over the distribution $\Unif(H_n)$, then the PRFF family would not be secure. This uses (1) that restricted to points within $H_n$, our concept class is equivalent to $\mcF$ and (2) that $H_n$ contains superpolynomially many points and the quantity $\alpha$ in \Cref{lem:PRFF-to-hardnes-of-learning} can be upper bounded by $\frac{m}{|H|}$. \qedhere

\begin{remark}[\Cref{thm:computational-body} relative to random oracle]
\Cref{thm:computational-body} has two assumptions, \Cref{conjecture:hard-to-sample} and the existence of one-way functions. Both assumptions are true relative to a random oracle: \Cref{conjecture:hard-to-sample} by \Cref{thm:hard-to-sample-relative-to-random-oracle} and existence of one-way functions because of the standard fact that random oracles are themselves one-way (with high probability). Combining these, we can make \Cref{thm:computational-body} unconditional relative to a random oracle, yielding a computational separation between standard PAC and sample PAC:
\end{remark}
\begin{theorem}[Computational separation of samplabe PAC and standard PAC relative to a random oracle]
    \label{thm:computational-sep-with-random-oracle}
    For every $n \in \N$, with $1 - o(1)$ probability over a random oracle $\bmcO$, there is a concept class $\mcC_{\bmcO}$ with the following properties.
    \begin{enumerate}
        \item Every concept in $\mcC_{\bmcO}$ is computed by a  size-$O(n^2)$ oracle circuit $A^{\bmcO}$.
        \item $\mcC_{\bmcO}$ is easy to learn in samplable PAC: For every $s \geq n, \eps > 0$, there is a time $\poly(s/\eps)$ learner (that doesn't even use the oracle $\bmcO$) that learns $\mcC_{\bmcO}$ to accuracy $1-\eps$ on all size-$s$ oracle generators $G^{\bmcO}$.     
        \item $\mcC_{\bmcO}$ is hard to learn in standard PAC: There is distribution under which learning $\mcC_{\bmcO}$ to constant accuracy requires exponential time, even against algorithms that have access to $\bmcO$.
    \end{enumerate}
\end{theorem}

\subsection{Proof of the computational analogue of~\Cref{thm:separations within samplable PAC intro}}
\label{sec:computational-within-samplable}

\begin{theorem}[Computational analogue of \Cref{thm:separations within samplable PAC intro}]
    \label{thm:computational analogue of separations within samplable PAC formal}
    For every $\delta>0$ and $t\ge n$, there exists a concept class $\mathcal{C}$ of $\poly(t)$-size circuits such that
    \begin{enumerate}
        \item Assuming the existence of one way functions, there is a size $(tn \log (t)/\delta)$-distribution over which no polynomial-time algorithm can learn $\mathcal{C}$ to constant error.
        \item For any distribution generated by a circuit of size $t\le s$ and $\eps  \geq 4\delta$, there is an efficient algorithm that learns $\mathcal{C}$ to error $\eps$ using $O(t\log (t) /\eps^2)$ many samples.
    \end{enumerate}
\end{theorem}
The proof of \Cref{thm:computational analogue of separations within samplable PAC formal} is essentially the same of \Cref{thm:computational-body}, so we just sketch the differences. The main difference is rather than using sets $H_n$ given by \Cref{conjecture:hard-to-sample}, it uses sets $H_n$ constructed in \Cref{lem:PRG evasive formal}. This concept class still contains only polynomial-sized circuits, because the sets constructed in \Cref{lem:PRG evasive formal} are explicit. It is also easy to learn over distributions generated by size $t \leq s$ circuits because the sets are evasive and so we can once again apply \Cref{lem:evasiveness implies learnability intro}.

For the hardness of learning over size $S \coloneqq (tn \log t/\delta)$-distributions, we need to show there is a size-$S$ distribution that approximately samples $\Unif(H)$. For this, we use rejection sampling: The sampler draws $O(1/\delta)$ many uniform elements of $\zo^n$ and returns the first falling within $H$ (each of these checks can be done by a circuit of size $O(tn \log t)$ by the explicitness of $H$). If none fall within $H$ (which occurs with probability $\leq 0.1$), then we just output an arbitrary $x^{\star} \in H$ that was memorized. This gives a size-$S$ distribution $\mcD$ outputs a uniform element of $H$ with probability $1-c$ and otherwise $x^{\star}$ for $c \leq 0.1$. We once again apply \Cref{lem:PRFF-to-hardnes-of-learning} (this time with the bound $\alpha \leq \frac{|H|}{m} +c$) which implies the hardness of learning over $\mcD$ to constant accuracy.

\section{Online learning against efficient adversaries}
\label{sec:online}

In this section we discuss the online analogue of samplable PAC---online learning against an {\sl efficient} adversary.
First, we formally define the online learning task and the adversary.
\begin{definition}[Online learning and mistakes]
An online learner
$A$
for a concept class $\mcC$ is an algorithm that 
receives an unlabeled input $x^{(t)}$ from an
adversary $\Adv$ at each round $t$ and outputs a label $y^{(t)}$.
The adversary then outputs a label $f(x^{(t)})$ such that the set of labeled examples $\{(x^{(u)}, f(x^{(u)})): u \le t\}$
is consistent with some function $f \in \mcC$.

The \emph{mistake bound} of $A$ with respect to $\Adv$ is
\[\mistakes(A, \Adv) = \sum_{t = 1}^\infty \Ind[y^{(t)} \ne f(x^{(t)})].\]

The mistake bound of the concept class $\mcC$ is 
\[\max_{\Adv}\min_A \{ \mistakes(A, \Adv)\}.\]
\end{definition}

Our aim is to show a separation between an inefficient adversary and an efficient one in terms of mistakes it can
produce in a learner.
\begin{definition}[Efficient adversary]
A size-$s$ adversary
for a concept class $\mcC$ is a size-$s$ circuit that runs every round.
The adversary maintains a state $\alpha$ and updates it each round; 
the state $\alpha^{(t)}$ is part of the input for round $t$,
and the state $\alpha^{(t+1)}$ is part of the output.
It also takes as input a label $y^{(t-1)} \in \zo$ representing the learner's last label, and outputs the true 
label $f(x^{(t-1)})$ and a new unlabeled input $x^{(t)}$.\footnote{For completeness, we will say that in every round the adversary acts first. In round one, the adversary takes no input and outputs $x^{(1)}$, and the learner takes only $x^{(1)}$ as input. The initial state for both is $\vec{0}$.}
The set of labeled examples $\{(x^{(u)}, f(x^{(u)}):u < t\}$ must be consistent with some $f \in \mcC$. 
\end{definition}

We also define an efficient learner, parameterized by the size of its memory.
\begin{definition}[Efficient online learner]
A size-$s$ learner for a concept class $\mcC$ is an algorithm that maintains an $s$-bit state $\alpha$ and runs in $\poly(s,n)$ time per round. 
It takes as input the adversary's label of the previous input, $f(x^{t-1})$, and 
a new unlabeled input $x^{(t)}$, and outputs a label $y^{(t)}$.
\end{definition}

To adhere to standard notions of efficiency in online learning, we represent the 
learner as a time-$s$ Turing machine. In contrast, the adversary is represented as a circuit, as it is the online analogue of the size-$s$ distribution in the PAC setting. Upper bounds in this setting are harder to come by than if we forced the adversary to be uniform or allowed the learner to be nonuniform, which only strengthens our results.



\subsection{Online analogue of~\Cref{thm:statistical intro}}

\begin{theorem}[Online analogue of \Cref{thm:statistical intro}]
\label{thm:statistical-online}
There exists a concept class $\mcC$ over $\zo^n$ such that:

\begin{itemize}
\item[(i)] 
There exists an adversary $\Adv$ such that for any learner $A$,
\[\mistakes(A, \Adv) \ge 2^n/2.\]
\item[(ii)] For any $s \ge n$ and adversary $\Adv_s$ of size $s$ for $\mcC$, there exists a size-$O(ns \log s)$ 
learner $A_s$ such that
\[\mistakes(A_s, \Adv_s) = O(s \log s).\]
\end{itemize}
\end{theorem}

Throughout this section and others, we will refer to this simple deterministic learner.
It outputs 0 by default unless it has seen the input before, and can store 
$O(s/n)$ previously-seen distinct examples in its $s$-bit memory.
It is the online analogue of the PAC learner from \Cref{lem:evasiveness implies learnability intro}.
\begin{definition}[Default-zero learner of state size $s$]
The learner's state $\alpha$ is a queue of $O(s/n)$ inputs.
On input $f(x^{(t-1)}), x^{(t)}$, it does the following:
\begin{enumerate}
    \item \textbf{Update:} If $f(x^{(t-1)}) = 1$ and $x^{(t-1)}$ is not in the queue, add $x^{(t-1)}$ to the queue and remove the oldest entry if the queue is full.
    \item \textbf{Output:} If $x^{(t)}$ is in the queue, output $y^{(t)} = 1$. Otherwise, output $y^{(t)} = 0$.
\end{enumerate}
\end{definition}
\begin{proof}[Proof of \Cref{thm:statistical-online}]

First we prove that a random $\bH$ satisfies $(i)$ with probability 1/2.
Let $\bH \sse \zo^n$ be a subset generated by including each $x$ independently with probability $1/2$. 
Consider $\mcF_{\bH}$, the set of 
all functions restricted to $\bH$ (as in \Cref{subsec:overview}).
The concept class $\mcC$ referred to in \Cref{thm:statistical-online} will be $\mcF_H$ for one such 
$H$, which will be shown to exist by the probabilistic method.
For any $H$, the VC dimension of $\mcA_H$ is $|H|$, as $\mcA_H$ shatters $H$.
We recall the well-known fact that the Littlestone dimension is at least the VC dimension, and that the Littlestone dimension characterizes mistake bound:
\[\max_{\Adv}\min_A\{\mistakes(A,\Adv)\} = \mathrm{LDim}(\mcA_H) \ge \VCdim(\mcA_H) = |H|.\]
We have $\Prx_{\bH}[|\bH| \ge 2^n/2] \geq 1/2$ by symmetry of the distribution $\mathrm{Bin}(2^n, 1/2)$, which is the distribution of $|\bH|$.

Now we prove that a random $H$ satisfies (ii) with probability 4/5.
Let $s$ be the size bound of the adversary.
Let the learner $A$ be the default-zero learner of size $s$.
The learner is deterministic, so fixing the adversary $\Adv$ 
fixes the sequence of possible labeled examples $(x^{(t)}, f(x^{(t)}))$
that the adversary can output when running with the learner $A$.
We will associate with this adversary a set $T_{\Adv}$, where $x \in T_{\Adv}$ if the adversary 
ever outputs $(x, 1)$. 
The adversary is \emph{valid} for $\mcA_H$ iff $T_{\Adv} \subseteq H$.

Then over the randomness of $\bH$, we have 
\[\Prx_{\bH}[\Adv\text{ is valid for $\mcA_{\bH}$}] = 2^{-|T_{\Adv}|}.\]
There are $s^{O(n+s)}$ Boolean circuits of size $s$, so we have 
\[\Prx_{\bH}\Brac{\text{$\exists$ valid size-$s$ adversary for $\mcA_{\bH}$ such that $|T_{\Adv}| \ge k$}} \le 2^{-k} \cdot s^{O(n+s)}.\]
Then there is some constant $c$ such that for $k \ge c(s+n) \log (s) $, we have 
\[\Prx_{\bH}\Brac{\text{$\exists$ valid size-$s$ adversary for $\mcA_{\bH}$ such that $|T_{\Adv}| \ge k$}} \le 15 n^{-2} \cdot 2^{-n}.\]
We now union bound over all choices of $s \le n^2\cdot 2^n$; it is not necessary to handle adversaries
larger than this because the learner can be as large as the adversary,
and the amount of memory required for the learner to store all $2^n$ points in memory is $O(n \cdot 2^n)$.
We then have 
\[\Prx_{\bH}\Brac{\text{$\exists$ valid size-$(n \cdot 2^n)$ adversary for $\mcA_{\bH}$ such that $|T_{\Adv}| \ge k$}} < \lfrac 15.\]
Thus for at least $4/5$ of possible choices of $H$, all valid size-$s$ adversaries
label at most $k(s) = O((s+n)\log s)$ entries with 1.
A default-0 learner with $kn$ space, then, has space to store all such entries and thus 
only makes a mistake the first time each of them is seen. Thus the mistake bound is $O((s+n)\log s)$.

To conclude the proof, since more than $4/5$ of possible 
choices of $H$ satisfy condition (ii) and $1/2$ satisfy condition (i),
there is a choice of $H$ that satisfies both, and the theorem holds
for the corresponding $\mcF_H$.
\end{proof}

\subsection{Online analogue of~\Cref{thm:separations within samplable PAC intro}}
\begin{theorem}[Online analogue of \Cref{thm:separations within samplable PAC intro}]
\label{thm:thm3-online}
For every $s \ge n$, there exists a concept class $\mcC_s$ of  such that:
\begin{itemize}
\item[(i)] There exists an adversary $\Adv_S$ of size $S = s \log s \cdot \poly(n)$ such that for every learner $A$,
\[\mistakes(A, \Adv) = 2^n/4.\]
\item[(ii)] For any adversary $\Adv_s$ of size $s$, there exists a size-$O(ns \log s)$ learner $A_s$ such that
\[\mistakes(A_s, \Adv_s) \le O(s \log s).\]
\end{itemize}
\end{theorem}	

%

\begin{proof}
As in \Cref{thm:separations within samplable PAC intro}, 
the set $\bH_s$ will be generated pseudorandomly from a $1/2$-biased $d$-wise independent PRG
(see \Cref{fact:explicit prgs for rvs}), for $d = \Theta(s \log s)$.
The concept class will be the class $\mcA_{H_s}$ for some $H_s$ in the support of $\bH_s$; we will build an
efficient adversary for it.

We cite from \Cref{fact:explicit prgs for rvs} that the size of the circuit computing the generator for $\bH_s$ is $S = O(n s \log s)$. 
We define $\Adv_S$ to do the following:
\begin{itemize}
\item For $t \le 2^n$, let $x^{(t)}$ be the binary expansion of $t$ (the adversary keeps a counter and stops after $2^n$ rounds).
\item On round $t$, read the input $y^{(t)}$ from the learner. Let
\[f(x^{(t)}) = \begin{cases} 0 &\text{if\ } x^{(t)} \not\in H_s \\
\lnot\, y^{(t)} &\text{if\ } x^{(t)} \in H_s
\end{cases}.\]
\item Output $f(x^{(t)})$ and $x^{(t+1)}$ to the learner.
\end{itemize}
The circuit size of the $n$-bit counter is $\poly(n)$, so the total size of the adversary is $s \log s \cdot \poly(n)$.
Clearly the mistake bound of any learner relative to this adversary is $\ge |H|$, since the adversary forces
a mistake on every member of $H$.
To show that $|H|$ is sufficiently large we cite \Cref{cor:cantelli} with 
$\delta = 1/2$, which gives that $|\bH_s| \ge 2^n/4$ with probability strictly greater than 1/5.
Therefore (i) holds with probability $> 1/5$ over the distribution of $\bH_s$.

The proof that (ii) holds with probability $> 4/5$ is nearly identical to that of (ii) in 
\Cref{thm:statistical-online}.
The only difference is that for $d$-wise independent $\bH_s$, we have 
\[\Prx_{\bH_s}[\Adv\text{ is valid for }\mcA_{\bH_s}] \le 2^{-\min(|T_{\Adv}|, d)}\]
instead of 
\[\Prx_{\bH_s}[\Adv\text{ is valid for }\mcF_{\bH_s}] \le 2^{-|T_{\Adv}|}.\]

Since $d = \Theta(s \log s)$, it is still the case that there is some constant $c$ such that for 
$k \ge c(s+n)\log s$,
\[\Prx_{\bH_s}\Brac{\text{$\exists$ valid size-$s$ adversary for $\mcA_{\bH_s}$ such that $|T_{\Adv}| \ge k$}} > \lfrac 45 n^{-2} \cdot 2^{-n}.\]
The rest of the proof proceeds identically to the proof of \Cref{thm:statistical-online}.
\end{proof}

\subsection{Online analogues of~\Cref{key assumption intro} and~\Cref{thm:computational intro}}
In this section, we prove a conditional computational separation between online learning with an efficient adversary versus an inefficient one.
We introduce a variation of the mistake bound, which we call the mistake rate.
\begin{definition}[Mistake rate]
The \emph{mistake rate} of $A$ with respect to an adversary $\Adv$ and a convergence time $k$ is defined as
\[\mrate_k(A, \Adv) = \max_{T \ge k} \Brac{\frac 1T \cdot \sum_{t = 1}^T \Ind[y^{(t)} \ne f(x^{(t)})]}.\]
\end{definition}
The mistake rate is a nonstandard way to measure the error of an online algorithm, but it is the natural
online analogue of an $(\eps, k)$-hitting distribution,
as the convergence time $k$ is determined by the number of points an efficient algorithm can memorize. 
Informally, we show that the adversary can force at most $\eps k$ mistakes due to memorization, and other than those $\eps k$ mistakes, will only be able to force an $\eps$-fraction of rounds to have a mistake.

To do so, we give an online analogue of \Cref{key assumption intro}. In place of size-$s$ distributions, we care about the sequence of points a size-$s$ online adversary can produce. 

\begin{conjecture}[Online analogue of \Cref{key assumption intro}]
\label{conj:online-assumption}
There exist sets $\{H_n \subseteq \zo^n\}_{n \in \N}$ satisfying the following.
\begin{enumerate}
\item Explicit: For every $n \in \N$, membership in $H_n$ is computed by a circuit of size $\poly(n)$.
\item Evasive: There is some constant $c$ such that for every polynomial $p(n)$ and sufficiently large $n$, 
the first $p(n)^{c+1}$ rounds produced by a size-$p(n)$ 
adversary with 
a default-zero learner of size $p(n)^c$
contain at most $p(n)^{c}$ distinct elements of $H_n$.
\item Large: $H_n$ has superpolynomial size.
\end{enumerate}
\end{conjecture}

For intuition, think of the adversary-learner pair as a circuit that generates a sequence of points.
In each round, the circuit updates a state and outputs some $x \in H_n$.
The conjecture claims that if the circuit has size $s$, then aside from being able to memorize $\approx s/n$ points, it should not be able to hit $H_n$ at a rate better than $1/s$ over a long sequence.

\begin{theorem}[Online analogue of \Cref{thm:computational intro}]
\label{thm:online-computational}
Assuming \Cref{conj:online-assumption} and that one-way functions exist,
there exists a concept class $\mcC$ of polynomial-size circuits such that:
\begin{itemize}
\item[(i)] There is an adversary $\Adv$ such that for all polynomials $p(n), q(n)$, for every learner $A$ running in time $p(n)$,
		\[\mrate_{q(n)}(A, \Adv) \ge 0.01.\]
\item[(ii)] For any adversary $\Adv_s$ of time $s = \poly(n)$ and $\eps > 0$, there exists a time-$(n \cdot \poly(s/\eps))$ learner $A_s$ such that
\[\mrate_{k}(A_s, \Adv_s) \le \eps.\]
for some $k = \poly(s/\eps)$.
\end{itemize}
\end{theorem}

The proof of this theorem makes use of a reduction from online learning to PAC learning.
The result of \cite{Lit89} relates the mistake bound to the sample complexity of PAC learning.
The relationship between the mistake rate, convergence time, and sample complexity,
as well as the relationship between the time complexities of online and PAC learning, are implicit in the proof.
\begin{fact}[Mistake rate learning implies PAC learning (implicit in \cite{Lit89})]
\label{fact:littlestone}
Let $\mcC$ be online-learnable with mistake rate $\eps$ at convergence time $k$ by a learner that runs in time $s$.
Then $\mcC$ is PAC-learnable with expected error $\le \tfrac 38 + 5\eps$ in time $O(sk + \frac{1}{\eps}\log(k))$ and sample complexity $O(k + \frac{1}{\eps}\log(k))$.
\end{fact}

\begin{proof}[Proof of \Cref{thm:online-computational}]
Let $\{H_n\}_{n \in \N}$ be the sets guaranteed to exist by the assumption of \Cref{conj:online-assumption}.
As in the proof of \Cref{thm:computational-body}, let $\mcF = \{f_s:\zo^{n} \to \zo\}_{s \in \zo^n}$ be a 
cryptographic PRFF with negligible security parameter $\gamma$, and let the concept class $\mcC_n$ 
consist of all 
\[c_s(x) = f_s(x) \wedge \Ind[x \in H_n] \]
for $s \in \zo^n$. 
We observe that each $c_s$ is a polynomial-size circuit; this follows from the fact that $f_s$ is a polynomial-size
circuit and the assumption that membership in $H_n$ is decided by a polynomial-size circuit as well.

\paragraph{Hardness of learning:} We will prove item (i) by reduction to \Cref{lem:PRFF-to-hardnes-of-learning}.
Suppose there is a time-$s$ online learner $A$ for $\mcC_n$ 
such that for any adversary,
$\mrate_k(A,\Adv) < \eps$. 
We will build a $\poly(s,k)$-time PAC learner for $\mcC_n$,
which will break PRFF security.
By \Cref{fact:littlestone}, there is a PAC learner for $\mcC_n$ with expected error $\le \tfrac 38 + 5\eps$ and $\poly(s,k)$ time and samples.
Then \Cref{lem:PRFF-to-hardnes-of-learning} gives a polynomial-time nonuniform adversary (of the PRF $f_s$) with advantage
$\frac{1}{8} - 5\eps - \frac{\poly(s,k)}{2|H_n|}$.
Substituting 0.01 for $\eps$ and applying the assumption that $|H_n|$ is superpolynomial
contradicts the assumption that $\gamma$ is negligible,
thus showing that such a learner does not exist.

\paragraph{Learnability with efficient adversaries:} To prove item (ii), 
if $s < 2/\eps$, we pad the adversary with inoperative gates so that its size becomes
$2/\eps$; if $s > 2/\eps$, we set $\eps$ to $2/s$. 
Then, break the sequence of rounds into blocks of length $s^{c+1}$, where $c$ is the constant given
by \Cref{conj:online-assumption}.
We will make a slight modification to the default-zero learner described in \Cref{thm:statistical-online}:
at the start of each block, it will empty its memory.
The learner will have space to store $s^{c}$ elements, and thus a space bound of $O(s^{c} \cdot n)$.
Observe that if there is a size-$s$ adversary that produces a sequence $x^{(1)},\ldots,x^{(s^c)}$ in a block, then there is another adversary of size $O(s)$ that produces the same sequence in the first block,
as the state of the first adversary at the start of the block can be hardcoded.
Then by \Cref{conj:online-assumption}, no block can contain more than $s^{c}$ distinct elements of $H_n$.
Since all functions in the concept class take value zero outside of $H_n$, and each of the $s^{c}$ 
elements is stored after its label is revealed, the learner can only make a mistake the first time each 
of these $s^{c}$ elements is shown. Thus, the fraction of rounds in this block on which the learner 
makes a mistake is at most $\eps/2$.
Since every block contains at most $\eps/2$ fraction of mistakes, after the first block is
finished, the running mistake rate is at most $\eps$.
Thus we have
\[\mrate_{s^c}(A, \Adv) \le \eps.\]
Since we set $s$ to the larger of $s$ and $2/\eps$, the size bound of the learner is
$O(n \cdot \poly(s/\eps))$ and the convergence time is $\poly(s/\eps)$.
\end{proof}
 
\subsection{Relating evasive-set conjectures}
In this section, we relate \Cref{conj:online-assumption} to the main evasive-set conjecture (\Cref{conjecture:hard-to-sample}) and the hardness of finding many satisfying assignments.
\begin{enumerate}
    \item We show that our online analogue of the evasive-set conjecture, \Cref{conj:online-assumption}, is true relative to a random oracle.
    The proof is analogous to the proof of \Cref{thm:hard-to-sample-relative-to-random-oracle}.
    \item We also show that  \Cref{conj:online-assumption} implies the offline variant \Cref{conjecture:hard-to-sample}.
\end{enumerate}


\begin{lemma}[\Cref{conj:online-assumption} is true relative to a random oracle]
\label{lem:online-random-oracle}
    For every sufficiently large $n$, with probability at least $1 -o(1)$ over a uniform random oracle $\bmcO$, there is a set $H_{\bmcO}$ satisfying:
    \begin{enumerate}
        \item Explicit: There is an oracle circuit $C$ of size $O(n^2)$ s.t. $C^{\bmcO}$ computes membership in $H_{\bmcO}$.
        \item Evasive: For every polynomial $p(n)$, oracle adversary $\Adv$ of size $p(n)$, and sufficiently large $n$, the first $p(n)^{2.01}$ rounds produced by $\Adv^{\bmcO}
        $ with a default-zero learner of size $p(n)^{1.01}$ contain at most $p(n)^{1.01}$ distinct elements of $H_n$.
        \item Large: $H_{\bmcO}$ accepts $2^{\Omega(n)}$ many unique inputs.
    \end{enumerate}
\end{lemma}
\begin{proof}
We will let $H_{\bmcO}$ be the set of accepting assignments to $C_{\bmcO}$ as defined in the proof of \Cref{lem:oracle-LB} with $\delta \coloneqq 2^{-0.5n}$.
Clearly this satisfies the explicitness condition.
By \Cref{lem:oracle-LB}, with probability at least $1 - \exp(-\Omega(2^{0.5n}))$,
$H_{\bmcO}$ has size at least $2^{0.5n}$, thus satisfying the largeness condition as well.
We will now argue that with high probability over $\bmcO$, $H_{\bmcO}$ also satisfies the evasiveness condition.

Fix an oracle adversary $\Adv$ of size $s \coloneqq p(n)$. 
Consider the algorithm that simulates the run of $\Adv$ with the default-zero learner of size $s^{1.01}$
and outputs all the points produced by $\Adv$. 
If $\Adv$ produces more than $s^{1.01}$ distinct elements of $H_{\bmcO}$ in the first $s^{2.01}$ rounds with the default-zero learner, then
this simulation is a query algorithm that outputs $s^{1.01}$ distinct elements of $H_{\bmcO}$ and makes $s^{3.01}$ oracle queries (since in each round, the adversary makes at most $s$ queries and the learner makes none. 
We apply \Cref{lem:oracle-LB} with $q = s^{3.01}$, $k = s^{1.01}$, and $\delta = 2^{-0.5n}$,
which shows that the success probability of this algorithm is at most $\exp(-\Omega(s^{1.01}))$ over the randomness of $\bmcO$.
Thus with probability at least $1 - \exp(-\Omega(s^{1.01}))$, this adversary does \emph{not} produce $s^{1.01}$ distinct elements of $H_{\bmcO}$.

We now union bound over the $\exp(O(s \log s))$ possible adversaries of size $s$.
Then, with probability at least $1 - \exp(O(s \log s - s^{1.01}) = 1 - o(1)$ over the randomness of $\bmcO$, 
no size-$s$ adversaries produce $s^{1.01}$ distinct elements of $H_{\bmcO}$.
This concludes the proof of the evasiveness condition.
\end{proof}

Now we relate the online and offline evasive-set conjectures.
\begin{lemma}
\label{lem:relate-2}
\Cref{conj:online-assumption} implies \Cref{conjecture:hard-to-sample}.
\end{lemma}
We first show the following supporting claims:
\begin{proposition}
    \label{prop:hitting-to-expected-pairwise}
    Let $\mcD$ be a distribution that $(\eps, k)$-hits a set $H$. Then, for $m \coloneqq \floor{2k/\eps}$ and any pairwise independent random variables $\bx^{(1)}, \ldots, \bx^{(m)}$ that are marginally from $\mcD$, meaning for all $i \neq j$ the marginal distribution of $(\bx_i, \bx_j)$ is that of $\mcD^2$,
    \begin{equation*}
        \Ex\bracket[\bigg]{\, \sum_{x \in H} \Ind[x \in \set{\bx^{(1)}, \ldots, \bx^{(m)}}]} \geq \frac{k}{4}.
    \end{equation*}
\end{proposition}
\begin{proof}
    Let $\bx^{(1)}, \ldots, \bx^{(m)}$ be pairwise independent and $\by^{(1)}, \ldots, \by^{(m)}$ be fully independent, both with the same marginal distribution of $\mcD$ on all $m$ coordinates. We will first show that for any $x$
    \begin{equation*}
        \Prx\bracket*{x \in \set{\bx^{(1)}, \ldots, \bx^{(m)}}} \geq \frac{1}{2} \Prx\bracket*{x \in \set{\by^{(1)}, \ldots, \by^{(m)}}}.
    \end{equation*}
    For each $i \in [m]$ let $\bz_i \coloneqq \Ind[\bx_i = x]$ and $\bZ = \sum_{i \in [m]} \bz_i$. Then, the $\bz_i$ are pairwise independent and each have expectation $p$. Therefore, denoting $p \coloneqq \mcD(x)$,
    \begin{equation*}
        \Ex[\bZ] = pm \quad\quad\text{and}\quad\quad\Var[\bZ]  \leq pm.
    \end{equation*}
    Applying the second moment method,
    \begin{align*}
         \Pr\bracket[\Big]{x \in \set[\Big]{\bx^{(1)}, \ldots, \bx^{(m)}}} &= \Pr[\bZ > 0] \\
         &\geq \frac{\Ex[\bZ]^2}{\Ex[\bZ^2]} \\
         &\geq \frac{p^2m^2}{p^2m^2 + pm} \\
         &= \frac{pm}{pm + 1}\\
         & \geq \min\set*{\frac{1}{2}, \frac{pm}{2}}.
    \end{align*}
    In contrast, for fully independent samples,
    \begin{equation*}
        \Pr\bracket[\Big]{x \in \set[\Big]{\by^{(1)}, \ldots, \by^{(m)}}} = 1 - (1 - p)^m \leq \min\set{1, pm}.
    \end{equation*}
    Therefore by linearity of expectation, we have
    \begin{equation*}
        \Ex\bracket*{\sum_{x \in H} \Ind[x \in \set{\bx^{(1)}, \ldots, \bx^{(m)}}]} \geq \frac{1}{2} \Ex\bracket*{\sum_{x \in H} \Ind[x \in \set{\by^{(1)}, \ldots, \by^{(m)}}]}.
    \end{equation*}
    From \Cref{prop:hitting-to-many}, it follows that 
    \[\Ex_{\by^{(1)},\ldots,\by^{(m)} \sim \mcD} \bracket[\bigg]{\,\sum_{x \in H} \Ind[x \in \set{\by^{(1)}, \ldots, \by^{(m)}}} \ge k/2;\]
    thus the desired result follows.
\end{proof}

\begin{proposition}
\label{prop:imp2}
If $H$ is $(\eps, k)$-hit by a size-$s$ distribution $\mcD$, then it is $(\eps/8, m = \floor{\frac{2k}{\eps}})$-hit by a size-$O(s \log s + \log m)$-adversary.
\end{proposition}
We proceed in a manner similar to the Let $C:\zo^{\ell} \to \zo^n$ be the size-$s$ generating circuit for $\mcD$. Let $\mcG: \zo^{O(\ell)} \to \zo^{s m}$ be the pairwise-independent generator of size $O(\ell \log \ell)$ given by \Cref{fact:explicit prgs for rvs}.
Our adversary will be hardcoded with a seed $r \in \zo^{O(\ell)}$ and keeps a counter of the round number.
On each round $t$, it does the following: 
\begin{itemize}
\item Increment the counter $t$.
\item Let $x^{(t)} = C(\mcG(r)^{(t)})$ and $f(x^{(t-1)}) = 0$.
\item Send $x^{(t)}$ and $f(x^{(t-1)})$ to the learner.
\end{itemize}
Since the size of $C$ is at most $s$, we can upper bound $\ell \leq s$ in which case the size of $\mcG$ is at most $O(s \log s)$. We also have that $|r| \leq O(s)$ and that the counter has size $O(\log m)$, so the total size of the adversary is $O(s \log s + \log m)$, as desired.

 Now we argue that there is some seed $r$ for which it produces many members of $H$.
Consider the distribution of 
\[\mcG(\br) \coloneqq \ba^{(1)},\ldots,\ba^{(m)}\] 
where $\br$ is uniform over $\zo^{O(\ell)}$.
Then $C(\ba^{(1)}),\ldots,C(\ba^{(m)})$ are pairwise independent and marginally distributed according to $\mcD$. 

By \Cref{prop:hitting-to-expected-pairwise}, we have
\begin{equation*}
\Ex_{\br \sim \zo^{s}}\bracket[\bigg]{\, \sum_{x \in H} \Ind\bracket*{x \in S(\br)}} \geq \frac{k}{4} \quad\quad\text{where}\quad S(r) \coloneqq \set*{C_{\phi}\paren[\big]{\mcG(r)^{(1)}}, \ldots, C_{\phi}\paren[\big]{\mcG(r)^{(m)}}} .
\end{equation*}
Therefore, there exists some $r$ for which the sequence of $x^{(t)}$'s produced by the learner contains at least $k/4$ members of $H$. This sequence has length $m = \floor{2k/\eps}$, and so at least $\frac{k/4}{\floor{2k/\eps}} \geq \eps/8$ fractions of the point in the sequence are unique members of $H$.

\begin{proof}[Proof of \Cref{lem:relate-2}]
We claim that the sets $H_n$ conjectured to exist by \Cref{conj:online-assumption} are the same sets that witness \Cref{conjecture:hard-to-sample}.

Suppose some $H$ does not witness \Cref{conjecture:hard-to-sample}: then for all constants $c$, there is some $s = \poly(n)$ such that $H$ is $(1/s, s^c)$-hit by some distribution generated by a size-$s$ circuit $C$.
We claim that for $c' = c/2$, there is some $s' = \poly(n)$ such that $H$ is $(1/s',(s')^{c'+1})$-hit by a size-$s'$ adversary.
By the bijection between $c$ and $c'$, it follows that for all $c'$, there is some $s' = \poly(n)$ such that $H$ is $(1/s',(s')^{c'+1})$-hit by a size-$s'$ adversary, and so $H$ does not witness \Cref{conj:online-assumption}.

Consider such an $H$ that does not witness \Cref{conjecture:hard-to-sample}.
By \Cref{prop:imp2}, $H$ is $(1/8s, 2s^{c+1})$-hit by a size-$O(s \log s)$-adversary, so this adversary produces $s^c/4$ distinct elements of $H$ in the first $2s^{c+1}$ rounds.
For ease of analysis we pad the adversary with inoperative gates so that its size is $s' = 2s^2$.
Then the number of elements produced is $(s')^{c/2}$, and the number of rounds is 
\[2s^{c+1} =8s^{(c+1)/2} \le (s')^{c/2+1},\]
as we can assume $n \geq 4$ in which case $s \ge 4$ and $s' \geq 64$.
Then $H$ is $(1/s', (s')^{c/2+1})$-hit by this adversary.
By the definition of evasion of adversaries, this set fails to witness \Cref{conj:online-assumption}.

Thus, since every set that witnesses \Cref{conj:online-assumption} also witnesses \Cref{conjecture:hard-to-sample}, the lemma follows.
\end{proof}

\section*{Acknowledgments}

We thank the anonymous reviewers for helpful comments and feedback.  

 Guy, Caleb, Carmen, and Li-Yang are supported by NSF awards 1942123, 2211237, 2224246, a Sloan Research Fellowship, and a Google Research Scholar Award. Guy is also supported by a Jane Street Graduate Research Fellowship and Carmen by an NSF GRFP. Jane is supported by NSF awards 2006664 and 310818 and an NSF GRFP.

\bibliography{ref}
\bibliographystyle{alpha}

\end{document}